\documentclass[journal, twocolumn]{IEEEtran}

\usepackage{amssymb} 
\usepackage{amsmath} 
\usepackage{color}
\usepackage{theorem} 
\usepackage{bm, bbm} 
\usepackage{color}

\IEEEtriggeratref{34}

\newtheorem{theorem}{\noindent Theorem}
\newtheorem{definition}[theorem]{\noindent Definition}

\newtheorem{lemma}[theorem]{Lemma}

\newtheorem{REMARK}[theorem]{\noindent Remark}
\newenvironment{remark}{\begin{REMARK}\rm}{\rm\end{REMARK}}
\newtheorem{EXAMPLE}[theorem]{\noindent Example}
\newenvironment{example}{\begin{EXAMPLE}\rm}{\rm\end{EXAMPLE}}
\renewcommand{\mathbf}[1]{{\bm{#1}}}     
\newenvironment{myalgorithm}{%
        \begin{minipage}{\columnwidth}\vspace{1.8ex}
        \makebox[0ex]{}\hrulefill\makebox[0ex]{}\\*%
             }{%
             \makebox[0ex]{}\hrulefill\makebox[0ex]{}\end{minipage}}
\AtBeginDocument{%
       \renewenvironment{proof}{\begin{IEEEproof}}{\end{IEEEproof}}
       
}
\newcommand{\GF}{{\mathrm{GF}}}
\newcommand{\bldzero}{{\mathbf{0}}}
\newcommand{\Epsilon}{{\mathcal{E}}}
\newcommand{\Alpha}{{\mathrm{A}}}
\newcommand{\Beta}{{\mathrm{B}}}
\newcommand{\Mu}{{\mathrm{M}}}

\newcommand{\blda}{{\mathbf{a}}}
\newcommand{\blde}{{\mathbf{e}}}
\newcommand{\bldu}{{\mathbf{u}}}

\newcommand{\uu}{\mathrm{T}}
\newcommand{\code}{{\mathcal{C}}}
\newcommand{\varcode}{{\mathsf{C}}}
\newcommand{\Code}{{\mathbb{C}}}
\newcommand{\Support}{{\mathsf{supp}}}
\newcommand{\Kout}{{\mathsf{K}}}
\newcommand{\Dout}{{\mathsf{D}}}
\newcommand{\Kin}{{\mathsf{k}}}
\newcommand{\Rin}{{\mathsf{r}}}
\newcommand{\Din}{{\mathsf{d}}}
\newcommand{\Hin}{{\mathsf{H}}}

\newcommand{\colspan}{{\mathsf{colspan}}}
\newcommand{\rank}{{\mathsf{rank}}}
\newcommand{\transpose}{{\mathsf{T}}}

\newcommand{\Prob}{{\mathsf{Prob}}}

\newcommand{\decoder}{{\mathcal{D}}}
\newcommand{\Tau}{{t}}
\newcommand{\Rho}{{r}}

\newcommand{\GRS}{{\mathrm{GRS}}}
\newcommand{\In}{{\mathrm{in}}}

\newcommand{\JJ}{{\mathcal{J}}}
\newcommand{\KK}{{\mathcal{K}}}
\newcommand{\LL}{{\mathcal{L}}}
\newcommand{\QQ}{{\mathcal{Q}}}
\newcommand{\RR}{{\mathcal{R}}}

\renewcommand{\mod}{{\mathrm{mod}}}

\newcommand{\Interval}[1]{{\langle #1 \rangle}}

\newcommand{\Ceil}[1]{{\lceil #1 \rceil}}
\newcommand{\GS}{{\mathit{\Theta}}}
\newcommand{\List}{L}

\newcommand{\eqed}{\hspace*{\fill}$\square$}

\newcommand{\hE}{\hat{E}}
\newcommand{\he}{\hat{e}}
\newcommand{\hJ}{\hat{J}}
\newcommand{\hJJ}{\hat{\JJ}}
\newcommand{\hLL}{\hat{\LL}}
\newcommand{\hS}{\hat{S}}
\newcommand{\hZ}{\hat{Z}}

\newcommand{\hEpsilon}{\hat{\Epsilon}}
\newcommand{\hLambda}{\hat{\Lambda}}
\newcommand{\hsigma}{\hat{\sigma}}
\newcommand{\hOmega}{\hat{\Omega}}
\newcommand{\homega}{\hat{\omega}}

\newcommand{\tcode}{\tilde{\code}}
\newcommand{\tH}{\tilde{H}}
\newcommand{\tS}{\tilde{S}}

\newcommand{\etal}{\emph{et al.}}

\newcommand{\ie}{\emph{i.e.}}

\title{Coding for Combined \\ Block--Symbol Error Correction}

\author{Ron M. Roth,~\IEEEmembership{Fellow,~IEEE,} and
        Pascal O. Vontobel,~\IEEEmembership{Senior Member,~IEEE}%
  \thanks{Manuscript received February~08, 2013;
          revised October~31, 2013;
          date of current version February 26, 2014.
          This paper was previously presented in part at the
          IEEE International Symposium on Information Theory,
          Istanbul, Turkey, July 2013.}%
  \thanks{R.~M.~Roth is with the
          Computer Science Department,
          Technion---Israel Institute of Technology,
          Haifa 32000, Israel.
          This work was done in part while visiting
          Hewlett--Packard Laboratories,
          1501 Page Mill Road,
          Palo Alto, CA 94304, USA
          (e-mail: ronny@cs.technion.ac.il).}%
  \thanks{P.~O.~Vontobel was with Hewlett--Packard Laboratories,
          1501 Page Mill Road, Palo Alto, CA 94304, USA.
          He is now with the 
          Department of Electrical Engineering, 
          Stanford University, Stanford, CA 94305, USA,
          and the
          Department of Information Technology
          and Electrical Engineering,
          ETH Zurich, 8092 Zurich, Switzerland
          (e-mail: pascal.vontobel@ieee.org).
          }
}

\begin{document}

\maketitle

\begin{abstract}
  We design low-complexity error correction coding schemes for channels 
  that introduce different types of errors and erasures:
  on the one hand, the proposed schemes can successfully deal with
  symbol errors and erasures, and, on the other hand, they can also
  successfully handle phased burst errors and erasures.
\end{abstract}

\begin{IEEEkeywords}
  Decoding,
  generalized Reed--Solomon (GRS) code,
  Feng--Tzeng algorithm,
  phased burst erasure,
  phased burst error,
  symbol erasure,
  symbol error.
\end{IEEEkeywords}

\section{Introduction}
\label{sec:introduction}

Many data transmission and storage systems suffer from different
types of errors at the same time. For example, in some data storage
systems the state of a memory cell might be altered by
an alpha particle that hits this memory cell. On the other hand,
an entire block of memory cells might become unreliable because of
hardware wear-out. Such data transmission and storage systems can be
modeled by channels that introduce symbol errors and
block (\ie, phased burst) errors,
where block errors encompass several contiguous
symbols. Moreover, if some side information is available, say based on
previously observed erroneous behavior of a single or of multiple memory
cells, this can be modeled as symbol erasures and block erasures.

In this paper, we design novel error correction coding schemes
that can deal with both symbol and block errors and both symbol
and block erasures for a setup as in
Fig.~\ref{fig:memory:block:1}.
\begin{itemize}

\item
Every small square corresponds to a symbol in
$F = \GF(q)$, where $q$ is an arbitrary prime power.
(In applications, $q$ is typically a small power of $2$.)

\item
All small squares are arranged in the shape of an $m \times n$
rectangular array.

\item
We say that a \emph{symbol error} happens if the content of a small
square is altered. We say that a \emph{block error} happens
if one or several small squares in a column of the array are altered.%
  \footnote{%
      In our setting, we think of the symbol errors and block errors
      as being caused by two different mechanisms.
      In this model, an observer cannot distinguish a block error from one or
      multiple symbol errors in the same column.
  }

\item
Similarly, we say that a \emph{symbol erasure} happens if the content
of a small square is erased and we say that
a \emph{block erasure} happens if all small squares in
a column of the array are erased.%
  \footnote{%
      The positions of the erased symbols and blocks are assumed to be
      provided as side information. Thus, the squares contain
      elements of $F$ (even at the erased positions), and some of
      the erased squares might in fact contain correct values.
    }

\end{itemize}

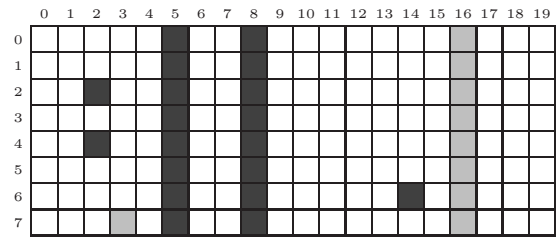
\begin{figure}
  \begin{center}
	\setlength{\unitlength}{2.2ex}
	\begin{picture}(21,9)(-1,0)
	\definecolor{error}{rgb}{.25,.25,.25}
	\definecolor{erasure}{rgb}{.75,.75,.75}
	\newsavebox{\error}\sbox{\error}{\color{error}\thicklines
		\multiput(.02,.02)(00,.02){49}{\line(1,0){.96}}
	}
	\newsavebox{\erasure}\sbox{\erasure}{\color{erasure}\thicklines
		\multiput(.02,.02)(00,.02){49}{\line(1,0){.96}}
	}
	\newcounter{x}

	\put(02,03){\usebox{\error}}
	\put(02,05){\usebox{\error}}
	\put(02,05){\usebox{\error}}
	\put(14,01){\usebox{\error}}
	\multiput(05,00)(0,1){08}{\usebox{\error}}
	\multiput(08,00)(0,1){08}{\usebox{\error}}
	\put(03,00){\usebox{\erasure}}
	\multiput(16,00)(0,1){08}{\usebox{\erasure}}

	\multiput(-.02,00)(0,1){09}{\line(1,0){20.04}}
	\multiput(00,00)(1,0){21}{\line(0,1){08}}

	\setcounter{x}{0}
	\put(.5,8.5){%
		\loop \ifnum \thex < 20
			\put(\thex,00){\makebox(0,0){\tiny$\thex$}}
			\addtocounter{x}{1}
		\repeat
	}
	
	\setcounter{x}{0}
	\put(-.5,7.5){%
		\loop \ifnum \thex < 8
			\put(00,-\thex){\makebox(0,0){\tiny$\thex$}}
			\addtocounter{x}{1}
		\repeat
	}
	\end{picture}
  \end{center}
  \caption{Array of size $m \times n$ with symbol errors/erasures 
    and block errors/erasures. Here, $m = 8$, $n = 20$,
    and there are symbol
    errors at positions $(2,2)$, $(4,2)$, and $(6,14)$,
    a symbol erasure at position $(7,3)$, block errors
    in columns $5$ and $8$, and a block erasure in column $16$.}
  \label{fig:memory:block:1}
  \vspace{-0.15cm}
\end{figure}

We can correct such errors and erasures by imposing that
the symbols in such an array constitute a codeword in some suitably
chosen code $\Code$ of length $mn$ over $F$.
The two main ingredients of the code $\Code$ that is proposed in
this paper are, on the one hand, a matrix $H_\In$
of size $m \times (mn)$ over $F$, and, on the other hand,
a code $\code$ of length $n$ over $F$. Namely, an array forms
a codeword in $\Code$ if and only if every row of the array is
a codeword in $\code$ once the $n$ columns have been transformed
by $n$ different bijective mappings $F^m \to F^m$ derived from
the matrix $H_\In$. The resulting error-correcting coding scheme has
the following salient features:

\begin{itemize}
\item
It can be seen as a concatenated coding scheme, however with two
somewhat distinctive
features. First, multiple inner codes
are used (one for every column encoding), and, second,
all these inner codes have rate one
  (\ie, the encoders of these inner codes can be considered to be
  column scramblers).

\item
One can identify a range of code parameters for $\Code$ for which (to
the best of our knowledge) the resulting redundancy improves upon
the best known.

\item
One can devise efficient decoders for 
combinations of symbol and block errors and erasures
most relevant in practical applications.
In particular, these decoders are more efficient than a corresponding
decoder for a suitably chosen generalized Reed--Solomon (GRS)
code of length $mn$ over $F$, assuming such a GRS code exists in
the first place. (Finding efficient decoders for the general
case is still an open problem.)
\end{itemize}

\subsection{Paper Overview}
\label{sec:paper:overview}

The paper starts in Section~\ref{sec:simplified:code:construction} by
considering a simplified version of the above error and erasure scenario
and of the above-mentioned code construction.
Namely, in this section we consider only block errors and erasures,
\ie, no symbol errors or erasures. Moreover,
an $m \times n$ array forms a codeword if and only if every row is
a codeword in some code $\code$ of length $n$
(\ie, there are no bijective mappings applied to the columns);
in other words, the array code considered is simply
an $m$-level interleaving of $\code$. Our main purpose of
Section~\ref{sec:simplified:code:construction} is laying out some of
the ideas and tools that will be used in subsequent sections;
in particular, it is shown how one can take advantage of
the \emph{rank} of the error array
in order to increase the correction capability of the array code.
Nevertheless, the discussion in
Section~\ref{sec:probabilistic} may be of independent interest
in that it provides a simplified analysis of
the decoding error probability of interleaved
GRS codes when used
in certain (probabilistic) channel models.

We then move on to
Section~\ref{sec:main:code:construction}, which is the heart of
the paper and which gives all the details of the above-mentioned code
construction and compares it with other code constructions. 
Finally,
Section~\ref{sec:decoding:algorithms} discusses a variety of decoders
for the proposed codes.

\subsection{Related Work}
\label{sec:related:work}

The idea of exploiting the rank of the error array
when decoding interleaved codes was presented
by Metzner and Kapturowski in~\cite{MK} and
by Haslach and Vinck in~\cite{HV1}, \cite{HV2}.
Therein, the code $\code$ is chosen to be a linear
$[n,k,d]$ code over $F$, and, clearly, any combination of
block errors can be corrected as long as
their number does not exceed $(d-1)/2$. In~\cite{MK} and~\cite{HV1},
it was further assumed that the set of nonzero columns in
the (additive) $m \times n$ error array $E$ over $F$ is
linearly independent over $F$;
namely, the rank of $E$ (as a matrix over $F$) equals the number
of block errors. It was then shown that under this additional
assumption, it is possible to correct (efficiently) any pattern of
up to $d - 2$ block errors.
Essentially, the linear independence allows to easily locate
the nonzero columns in $E$, and from that point onward, the problem
reduces to that of erasure decoding.
A generalization to the case where the nonzero columns in $E$
are not necessarily full-rank was discussed in~\cite{HV2};
we will recall the latter result in mode detail
in Section~\ref{sec:phased}.

The case where the constituent code $\code$ is a GRS code
has been studied in quite a few papers,
primarily in the context where the contents of each block error
is assumed to be uniformly drawn from $F^m$.
In~\cite{BKY}, Bleichenbacher~\etal\
identified a threshold, $(m/(m{+}1))(d{-}1)$,
on the number of block errors, below which the decoding failure
probability approaches~$0$ as $d$ goes
to infinity and $n/q$ goes to~$0$.
A better bound on the decoding error probability was obtained by
Kurzweil~\etal~\cite{KSH} and by Schmidt~\etal~\cite{SSB0, SSB}.
See also
Brown~\etal~\cite{BMS},
Coppersmith and Sudan~\cite{CS},
Justesen~\etal~\cite{JTH},
Krachkovsky and Lee~\cite{KL},
and Wachter--Zeh~\etal~\cite{WZZB}.

Turning to the main coding problem studied in this paper---%
namely, handling combinations of symbol errors and block errors---%
a general solution was given by
Zinov'ev~\cite{Zinovev} and Zinov'ev and Zyablov~\cite{ZZ},
using concatenated codes and
their generalizations. Specifically, 
when using an (ordinary) concatenated code,
the columns of the $m \times n$ array are set to be codewords of
a linear $[m,\Kin,\Din]$ inner code over $F$,
and each of these codewords is the result of
an encoding of a coordinate of an outer codeword
of a second linear $[n,\Kout,\Dout]$ code over $\GF(q^\Kin)$.
It follows from the analysis in~\cite{Zinovev} and~\cite{ZZ}
that any error pattern of up to $\vartheta$ symbol errors
and $\tau$ block errors can be correctly decoded, whenever
\begin{align*}
    2 \vartheta + 1 
      &\le
         \Din (\Dout - 2 \tau) \; .
\end{align*}
Furthermore, such error patterns can be efficiently decoded,
provided that the inner and outer codes have
efficient bounded-distance error--erasure decoders.

Note that (in the nontrivial case) when $\vartheta > 0$,
the rate of the inner code must be (strictly) smaller than $1$.
This, in turn, implies that the overall redundancy of
the concatenated code has to grow (at least) linearly with $n$.
A generalized concatenated (GC) code allows
to circumvent this impediment.
We briefly describe the approach, roughly following
the formulation of Blokh and Zyablov~\cite{BZ}.
Given a number $\tau$ of block errors and a number $\vartheta$
of symbol errors that need to be corrected, let the integer sequences
\begin{alignat*}{3}
  1
    &=\,&\Din_0
    &< \Din_1
    &< \cdots
    &< \Din_v \; , \\
    &&
  \Dout_0
    &\ge
       \Dout_1
    &\ge
       \cdots
    &\ge
       \Dout_v
\end{alignat*}
satisfy, for every $i = 0, 1, \ldots, v$,
\begin{align}
\label{eq:generalizedconcatenated}
    2 \vartheta + 1 
      &\le
         \Din_i (\Dout_i - 2 \tau)
\end{align}
(thus, $\Dout_0 \ge 2(\tau+\vartheta) +1$ and $\Dout_v \ge 2\tau + 1$).  For
$i = 1, 2, \ldots, v$, let $\Hin_i$ be an $\Rin_i \times m$ parity-check
matrix of a linear $[m,\Kin_i{=}m{-}\Rin_i,\Din_i]$
code over $F$. We further
assume that these codes are (strictly) nested, so that $\Hin_{i-1}$ forms a
proper $\Rin_{i-1} \times m$ sub-matrix of $\Hin_i$; the matrix formed by the
remaining $\Rin_i - \Rin_{i-1}$ rows of $\Hin_i$ will be denoted by $\partial
\Hin_i$ (we formally define $\Rin_0 = 0$, along with setting $\Hin_0$ to be an
``empty'' $\Rin_0 \times n$ matrix). Then an $m \times n$ array over $F$ is a
codeword of the generalized concatenated code, if and only if the following
two conditions hold:
\begin{list}{}{\settowidth{\labelwidth}{\textup{(G2)}}%
               \settowidth{\leftmargin}{\textup{(G2..)}}}
\item[\textup{(G1)}]
For $i = 1, 2, \ldots, v$, the (partial) syndromes 
of the columns with respect to the partial parity-check matrix
$\partial \Hin_i$ form a codeword of a code of length
$n$ over $F^{\Rin_i-\Rin_{i-1}}$
with minimum distance $\Dout_{i-1}$.
\item[\textup{(G2)}]
The columns of the array form a codeword of a code of length
$n$ over $F^m$ with minimum distance $\Dout_v$.
\end{list}
(Note that condition~(G2) could be incorporated into condition~(G1)
by extending the latter to $i = v+1$, with $\Hin_{v+1}$
taken as an $m \times m$ (nonsingular)
parity-check matrix of the trivial code $\{ \bldzero \}$.
Ordinary concatenated codes correspond to the case
where $v = 1$ and $\Dout_0$ is the ``minimum distance''~($> n$)
of the trivial code.)
It follows from~\cite{Zinovev} and~\cite{ZZ} that the above array code
construction has an efficient decoder that corrects any pattern
of up to $\tau$ block errors and up to $\vartheta$ symbol errors.
See also~\cite{AH},
\cite{KHSN},
\cite{RS},
\cite{WC},
and the survey~\cite{Dumer}.

Recently, Blaum~\etal~\cite{BHH} have proposed
new erasure-correcting codes for combined block--symbol
error patterns. The advantage of their scheme is having
the smallest possible redundancy
(equaling the largest total number of symbols that can be erased)
and an efficient \emph{erasure} decoding algorithm.
However, the parameters of their constructions are
rather strongly limited:
first, the array size is typically much smaller than $q$
(and, in one application,
must in fact be smaller than $\log_2 q$),
and, secondly, verifying whether the construction actually works
for given parameters becomes intractable, unless 
the number of block erasures or the number of symbol erasures
is very small.

In~\cite{GYD}, Gabrys~\etal\ presented a coding scheme which is targeted
mainly at applications for flash memories. In their setting, an erroneous
column may have at most a prescribed number $\ell$ of symbol errors; and in
addition to limiting the total number of erroneous columns, a further
restriction is assumed on the number of columns with at most a prescribed
number $\ell' \; (< \ell)$ of symbol errors.

In Section~\ref{sec:examples} we will compare our coding scheme with the most
relevant of the above-mentioned coding schemes.

\subsection{Notation}
\label{sec:notation}

This subsection lists the notation that will be used throughout the paper.
More specialized notation will be introduced later on when needed.

For integers $a$ and $b$ with $0 \leq a < b$,
we denote by $\Interval{a,b}$
the set of integers $\{ a, a + 1, a + 2, \ldots, b - 1 \}$,
and $\Interval{b}$ will be used as a shorthand notation
for $\Interval{0,b}$.  Entries of vectors will be indexed
starting at~$0$, and so will be the rows and columns of matrices.  For a
vector $\bldu \in F^n$ and a subset $W \subseteq \Interval{n}$, we let
$(\bldu)_W$ be the sub-vector (in $F^{|W|}$) of $\bldu$ that is
indexed by $W$.
The support of $\bldu$ will be denoted by $\Support(\bldu)$.
We extend these definitions to any $m \times n$ matrix $E$ over $F$,
with $(E)_W$ denoting the $m \times |W|$ sub-matrix of $E$
that is formed by the columns that are indexed by $W$.
Column~$j$ of $E$ will be denoted by $E_j$, and $\Support(E)$
will stand for the column support of $E$, namely, the set of
indexes $j$ for which $E_j \ne \bldzero$.
The linear subspace of $F^m$ that is spanned by the columns of $E$
will be denoted by $\colspan(E)$.

The ring of polynomials in the indeterminate $x$ over $F$ will be
denoted by $F[x]$, and
the ring of bivariate polynomials in $y$ and $x$ over $F$ will be
denoted by $F[y,x]$.\footnote{We prefer the ordering $y,x$ over $x,y$
because the powers of $y$ and the powers of $x$ will be
associated with, respectively,
the rows and columns of $m \times n$ matrices like $E$.}  For
a nonzero bivariate polynomial
$\varphi(y,x) = \sum_i \varphi_i(y) x^i$ in
$F[y,x]$, we will let $\deg_x \varphi(y,x)$ stand for
the $x$-degree of $\varphi(y,x)$, namely, the largest index $i$
for which $\varphi_i(y) \ne 0$.
The $y$-degree is defined in a similar manner.
The notation $F_{m,n}(y,x)$ will stand for
the set of all bivariate polynomials $\varphi(y,x) \in F[y,x]$
with $\deg_y \varphi(y,x) < m$ and $\deg_x \varphi(y,x) < n$.
For an element $\xi \in F$,
we denote by $\uu_m(y;\xi)$ the polynomial
$\sum_{i \in \Interval{m}} \xi^i y^i$.

With any $m \times n$ matrix
$E = (e_{h,j})_{h \in \Interval{m},j \in \Interval{n}}$ over $F$,
we associate the bivariate polynomial
\begin{align*}
E(y,x) &= \sum_{h \in \Interval{m},j \in \Interval{n}} e_{h,j} y^h x^j
\end{align*}
in $F_{m,n}(y,x)$ (namely, the powers of $y$ index the rows and
the powers of $x$ index the columns).
With each column~$j$ of $E$ we associate the univariate polynomial
$E_j(y) = \sum_{h \in \Interval{m}} e_{h,j} y^h$; thus,
$E(y,x) = \sum_{j \in \Interval{n}} E_j(y) x^j$.

\section{Simplified Code Construction}
\label{sec:simplified:code:construction}

In this section we consider the simplified scenario mentioned in
Section~\ref{sec:paper:overview}. Namely, we consider only
block errors and erasures, \ie, no symbol errors or erasures.
Moreover, an $m \times n$ array forms a codeword of length $mn$
if and only if every row is a codeword in some prescribed code $\code$
with parameters $[n,k,d]$ (\ie, there are no bijective mappings
applied to the columns); equivalently, the array code considered
is simply an $m$-level interleaving of $\code$.
If $\code$ is specified by
an $(n{-}k) \times n$ parity-check matrix $H$, 
then the syndrome matrix $S$ is defined to be
the $m \times (n{-}k)$ matrix $S = Y H^\transpose$,
where the $m \times n$ matrix
\begin{align*}
 Y  &= \Gamma + E
\end{align*}
over $F$ represents the read out (or received) message,
where the $m \times n$ matrix $\Gamma$ over $F$ represents
the stored (or transmitted) codeword, and where
the $m \times n$ matrix $E$ over $F$ represents the alterations that
happen to $\Gamma$ over time (or during transmission).
Note that our formalism treats erasures like errors,
with the side information $K \subseteq \Interval{n}$
telling us their location.

The subsections of this section are structured as follows.
In Section~\ref{sec:phased} we study the error correction capabilities
of the interleaved array code, where $\code$ is
any linear $[n,k,d]$ code over $F$.
Then, in Section~\ref{sec:GRS}, we present an efficient decoder
for the special case where $\code$ is a GRS code.
Finally, in Section~\ref{sec:probabilistic},
we present an application of the efficient decoder of
Section~\ref{sec:GRS} for the probabilistic decoding
of the array code under the assumption that the block errors
are uniformly distributed over $F^m$.

\subsection{Block Errors and Erasures with Rank Constraints}
\label{sec:phased}

We start with Theorem~\ref{thm:generalizedmetzner} below
that generalizes the results of~\cite{MK} and~\cite{HV1} to
the case where the set of nonzero columns of
the $m \times n$ error array $E$ are not necessarily linearly
independent. Note that this theorem was already stated
(without proof) in the one-page abstract~\cite{HV2} for
the error-only case (\ie, no block erasures).
We include the proof of the theorem not just for the sake
of completeness, but also because the proof technique will
be useful in Section~\ref{sec:main:code:construction} as well.

Toward proving this theorem, the following lemma will be helpful.

\begin{lemma}
\label{lem:generalizedmetzner}

Let $\code$ be a linear $[n,k,d]$ code over $F$ and
let $Z$ be a nonzero $m \times n$ matrix over $F$ such that each row
in $Z$ is a codeword of $\code$.  Then
  \begin{align*}
    \left|
      \Support(Z)
    \right|
    -
    \rank(Z)
      &\ge
         d - 1 \; .
  \end{align*}
\end{lemma}

\begin{proof}
Write $J = \Support(Z)$, $\Tau = |J|$, and $\mu = \rank(Z)$,
and apply the Singleton bound to the linear $[\Tau,\mu,{\ge}d]$ code
over $F$ that is spanned by the rows of $(Z)_J$.
\end{proof}

\begin{theorem}
\label{thm:generalizedmetzner}

Let $\code$ be a linear $[n,k,d]$ code over $F$ and let $H$ be
an $(n{-}k) \times n$ parity-check matrix of $\code$ over $F$.
Fix $K$ to be a subset of $\Interval{n}$ of size $\Rho$.
Given any $m \times (n{-}k)$ (syndrome) matrix $S$ over $F$,
there exists at most one $m \times n$ matrix $E$ over $F$ that
has the following properties:
\begin{list}{}{\settowidth{\labelwidth}{\textup{(ii)}}}

\item[\textup{(i)}]
$S = E H^\transpose$~, and

\item[\textup{(ii)}]
writing $\overline{K} = \Interval{n} \setminus K$, the values
$\Tau = \left|\Support \left( (E)_{\overline{K}} \right) \right|$ and
$\mu = \rank \bigl( (E)_{\overline{K}} \bigr)$ satisfy
  \begin{align}
    \label{eq:generalizedmetzner}
    2 \Tau + \Rho
      &\le
         d + \mu - 2 \; .
  \end{align}

\end{list}
\end{theorem}

\begin{proof}
We consider first the case where $K$ is empty. The proof is by
contradiction. So, assume that $E$ and $\hE$ are
two distinct $m \times n$ matrices over $F$
that satisfy conditions (i)--(ii).  Write
  \begin{align*}
    \Tau_{\max}
      &= \max
           \big\{
             |\Support(E)|, |\Support(\hE)|
           \big\} \; , \\
    \mu_{\max}
      &= \max
           \big\{ \rank(E), \rank(\hE) \big\}
             \; ,
  \end{align*}
  and define
  \begin{align*}
    J
      &= \Support(E) \; , \\
    \hJ
      &= \Support(\hE) \; , \\
    Q
      &= \Support(E) \cap \Support(\hE) \; .
  \end{align*}
Consider the array $Z = E - \hE$. By condition~(i) we get
that every row in $Z$ is a codeword of $\code$.  Now,
\begin{align*}
  |\Support(Z)|
    &\le
       |J| + |\hJ| - |Q| \; , \\
  \rank(Z)
    &\ge
       \mu_{\max} - |Q| \; ,
\end{align*}
and, so,
\begin{align*}
  |\Support(Z)| - \rank(Z)
    &\le |J| + |\hJ| - \mu_{\max} \\
    &\le 2 \Tau_{\max} - \mu_{\max} \\
    &\le d - 2 \; ,
\end{align*}
where the last inequality follows from condition~(ii). Hence, by
Lemma~\ref{lem:generalizedmetzner} we conclude that $Z = 0$,
namely, $E = \hE$, which is a contradiction to the initial assumption.
  
Next, we consider the case where $\Rho = |K| > 0$.  First, note that
condition~(ii) implies that
$\Rho \le d + \mu - 2 - 2 \Tau \le d - t - 2$;
in particular, every subset of $\Rho$ columns in $H$ is linearly
independent.
By applying elementary linear operations to the rows of $H$,
we can assume without loss of generality that the first $\Rho$ rows of
$(H)_K$ contain the identity matrix, whereas
the remaining $n{-}k{-}\Rho$
rows in $(H)_K$ are all-zero. Let $\tH$ be
the $(n{-}k{-}\Rho) \times (n{-}\Rho)$ matrix which consists of
the last $n - k - \Rho$ rows of $(H)_{\overline{K}}$:
the matrix $\tH$ is a parity-check matrix of
the linear $[n{-}\Rho,k]$ code $\tcode$ over $F$ obtained by
puncturing $\code$ on the positions that are indexed by $K$.
Let $\tS$ be the $m \times (n{-}k{-}\Rho)$ matrix which consists of
the last $n - k - \Rho$ columns of $S$. We have
  \begin{align}
    \label{eq:Stilde}
    \tS
      &= (E)_{\overline{K}} \tH^\transpose \; .
  \end{align}
Replacing~(i) by~(\ref{eq:Stilde}) and $\code$ by $\tcode$,
we have reduced to the case where $K$ is empty.
The result follows by recalling that the minimum distance of
$\tcode$ is at least $d-\Rho$.
\end{proof}

The proof of Theorem~\ref{thm:generalizedmetzner}
in~\cite{MK} and~\cite{HV1}, which was for the special case
$\Rho = 0$ and $\mu = \Tau = |\Support(E)|$, was carried out by
introducing an efficient algorithm for decoding up to
$\Tau \le d - 2$ errors.  In that algorithm,
Gaussian elimination is performed on the columns of $S$,
resulting in an $m \times (n{-}k)$ matrix $S P^\transpose$,
for some invertible $(n{-}k) \times (n{-}k)$ matrix $P$ over $F$,
such that the first $\mu$ columns in $S P^\transpose$ form
a linearly independent set while the last $n - k - \mu$ columns in $S
P^\transpose$ are all-zero. (As a matter of fact, through this Gaussian
elimination, one finds the value of $\mu$.) From condition~(i)
in Theorem~\ref{thm:generalizedmetzner}
we then get that
\begin{align}
  \label{eq:SP}
  S P^\transpose
    &= E (P H)^\transpose \; .
\end{align}
Let $H'$ be the $(n{-}k{-}\mu) \times n$ matrix that is formed by
the last $n - k - \mu$ rows of $P H$.
It follows from~(\ref{eq:SP}) that the columns of $H'$
that are indexed by $\Support(E)$ must be all-zero.
Furthermore, all the remaining columns in $H'$ must be nonzero,
or else we would have $\mu + 1 < d$ linearly dependent columns in $H$.
It follows that $\Support(E)$ is the unique subset
$U \subseteq \Interval{n}$ of size $\mu$ such that $(H')_U$ is
all-zero.  Once the decoder identifies $\Support(E)$,
the entries of $E$ can be found by solving linear equations.
This decoding algorithm can be generalized to handle
the case where $\Rho = |K| > 0$, in the spirit of the last part of
the proof of Theorem~\ref{thm:generalizedmetzner}: replace $E$,
$H$, and $S$ by $(E)_{\overline{K}}$, $\tH$, and $\tS$,
respectively.

As pointed out in~\cite{HV2}, when $K$ is empty and
the difference $\Tau - \mu = |\Support(E)| - \rank(E)$ is assumed
to be equal to some nonnegative integer $b$,
then the decoding algorithm of~\cite{MK} and~\cite{HV1} can be
generalized into finding a subset
$U \subseteq \Interval{n}$ of size $\Tau \le (d+\mu)/2 - 1$
such that $\colspan\bigl( (H')_U \bigr) \; (\subseteq F^{n-k-\mu})$
has dimension $b$.  Letting $V$ be a subset of $U$ of size $b$
such that $\rank\bigl( (H')_V \bigr) = b$,
the subset $W = U \setminus V$ will then point at
$\Tau - b = \mu$ columns of $E$ that form a basis of $\colspan(E)$;
these $\mu$ columns, in turn, are flagged as $\mu$ erasures.
We will then be left with
$|\Support(E) \setminus W| = \Tau - \mu = b$ nonzero columns in $E$
which are yet to be located, but these can be found by
applying to the received array, row by row,
any bounded-distance error--erasure decoder $\decoder$ for $\code$.
Indeed, since
\begin{align*}
  2 b + \mu
    &= 2 \Tau - \mu \le d-2 \; ,
\end{align*}
such a decoder can uniquely recover $E$ given the set $W$ of erasure
locations.  The extension to $\Rho = |K| > 0$ is straightforward.

Note, however, that even when we are allowed to apply
the decoder $\decoder$
at no computational cost, we do not know how to find the subset $U$
efficiently as $b$ becomes large.  In fact, if $b$ is large and $\mu$ is
small, we may instead enumerate over the set $W$ which indexes
a basis of $\colspan(E)$,
then use $\decoder$ to reconstruct a candidate for $E$, and finally
verify that we indeed have $\rank(E) = \rank\bigl( (E)_W \bigr) = \mu$.

\subsection{The GRS Case}
\label{sec:GRS}

In this section, we present an efficient decoder for finding
the error matrix $E$ under
the conditions of Theorem~\ref{thm:generalizedmetzner}, for the
special case where $\code$ is a generalized Reed--Solomon (GRS) code
(this decoder will then be used as a subroutine in one of the decoders
to be presented in Section~\ref{sec:decoding:algorithms}).
Specifically, hereafter in this section, we fix $\code$ to be an
$[n,k,d{=}n{-}k{+}1]$ GRS code $\code_\GRS$ over $F$ with
a parity-check
matrix
\begin{align*}
  H_\GRS
    &= \left( \, 
         \alpha_j^i \,
       \right)_{i \in \Interval{d-1}, j \in \Interval{n}} \; ,
\end{align*}
where $\alpha_0, \alpha_1, \ldots, \alpha_{n-1}$
are distinct nonzero elements of $F$
  (without real loss of generality, and for the sake of simplicity,
  we restrict ourselves here to GRS codes where the column
  multipliers---as defined in~\cite[p.~148]{Roth}---are all~$1$).

Let
$E = (e_{h,j})_{h \in \Interval{m}, j \in \Interval{n}}$ be
an $m \times n$ (error) matrix over $F$, and let $K$ and $J$ be
disjoint subsets of $\Interval{n}$ such that
\begin{align*}
  J
    &\subseteq
       \Support(E)
     \subseteq
       K \cup J
\end{align*}
(the set $J$ indexes the erroneous columns and $K$ indexes the erasure
locations). For the matrix $E$, define the syndrome array as
the following $m \times (d{-}1)$ matrix $S$ (equivalently,
the bivariate polynomial $S(y,x) \in F_{m,d-1}(y,x)$):
\begin{align*}
  S
    &= E H_\GRS^\transpose
\end{align*}
(compare with condition~(i) in Theorem~\ref{thm:generalizedmetzner}).
In addition, define the error-locator polynomial $\Lambda(x)$ and the
erasure-locator polynomial $\Mu(x)$ by
\begin{align*}
  \Lambda(x)
    &= \prod_{j \in J}
         (1 - \alpha_j x) \; , \\
  \Mu(x)
    &= \prod_{j \in K}
         (1 - \alpha_j x) \; ,
\end{align*}
respectively. Also, let the modified syndrome array be
the unique $m \times (d{-}1)$ array $\sigma$ over $F$ (namely,
the unique bivariate polynomial $\sigma(y,x)$ in
$F_{m,d-1}(y,x)$) that satisfies the polynomial congruence
\begin{align*}
  \sigma(y,x)
    &\equiv S(y,x) \, \Mu(x) \quad (\mod \; x^{d-1}) \; .
\end{align*}
Finally, the (bivariate) error-evaluator polynomial
$\Omega(y,x) = \sum_{h \in \Interval{m}} \Omega_h(x) y^h$ is defined by
\begin{align*}
  \Omega_h(x)
    &= \sum_{j \in K \cup J}
         e_{h,j}
         \prod_{j' \in (K \cup J) \setminus \{ j \}}
           (1 - \alpha_{j'} x) \; ,
             \quad h \in \Interval{m} \; .
\end{align*}

\begin{lemma}
\label{lem:GRS}

Write $\Tau = |J|$, $\Rho = |K|$, and $\mu = \rank\bigl( (E)_J \bigr)$, 
and suppose that
  \begin{align*}
    2 \Tau + \Rho
      &\le
         d + \mu - 2
  \end{align*}
(see~(\ref{eq:generalizedmetzner})).
Let $\lambda(x)$ be a polynomial in $F[x]$
and $\omega(y,x) = \sum_{h \in \Interval{m}} \omega_h(x) y^h$ be
a bivariate polynomial (of $y$-degree less than $m$) in $F[y,x]$
such that the following conditions are satisfied:
\begin{list}{}{\settowidth{\labelwidth}{\textup{(P2)}}%
               \settowidth{\leftmargin}{\textup{(P2..)}}}
\item[\textup{(P1)}]
  $\sigma(y,x) \lambda(x) \equiv \omega(y,x) \quad (\mod \; x^{d-1})$~,
  and

\item[\textup{(P2)}]
  \begin{minipage}[t]{\linewidth}
      $\deg \lambda(x)
        \le (d{+}\mu{-}\Rho)/2 - 1$ and \\
      $\deg_x \omega(y,x)
        <   (d{+}\mu{+}\Rho)/2 - 1$~.
    \end{minipage}

\end{list}
Then there is a polynomial $u(x) \in F[x]$ such that
\begin{align*}
    \lambda(x)
      &= \Lambda(x) u(x) \; , \\
    \omega(y,x)
      &= \Omega(y,x) u(x) \; .
\end{align*}
\end{lemma}

\begin{proof}
First, by the key equation of GRS decoding,
it is known that (P1)--(P2) are satisfied for
$\lambda(x) = \Lambda(x)$ and $\omega(y,x) = \Omega(y,x)$
(see, for example, \cite[Section~6.3 and pp.~207--208]{Roth}).
  
Now, let $\lambda(x)$ and $\omega(y,x)$ satisfy (P1)--(P2),
fix $\ell$ to be any index in $J$,
and let $U$ be a subset of $J$ of size $\mu$ such that $\ell \in U$
and $\rank\bigl( (E)_U \bigr) = \mu$ (recall that
$J \subseteq \Support(E)$ and, so,
$E_\ell \ne \bldzero$).
Let $\blda = (a_h)_{h \in \Interval{m}}$ be
a row vector in $F^m$ such that $(\blda E)_U$ is nonzero
on---and only on---position $\ell$.
  
Let $\hJ$ be the smallest subset of $\Interval{n}$ such that
\begin{align*}
    \hJ
      &\subseteq
         \Support(\blda E)
       \subseteq
         K \cup \hJ \; ;
\end{align*}
note that $\ell \in \hJ$ and that
$\hJ \subseteq \{ \ell \} \cup (J \setminus U)$ and, so,
$|\hJ| \le \Tau - \mu + 1$.  Define
\begin{align*}
    \hLambda(x)
      &= \prod_{j \in \hJ}
           (1 - \alpha_j x) \; , \\
    \hOmega(x)
      &= \frac{1}{\prod\limits_{j'' \in U \setminus \{ \ell \}}
         (1 - \alpha_{j''} x)}
         \cdot
         \sum_{h \in \Interval{m}}
           a_h \Omega_h(x) \\
      &= \sum_{j \in K \cup \hJ}
           (\blda E_j)
           \prod_{j' \in (K \cup \hJ) \setminus \{ j \}}
             (1 - \alpha_{j'} x) \; , \\
    \hsigma(x)
      &= \sum_{h \in \Interval{m}}
           a_h \sigma_h(x) \; , \\
    \homega(x)
      &= \sum_{h \in \Interval{m}}
           a_h \omega_h(x) \; .
\end{align*}
Observing that $\hsigma(x)$ is the modified syndrome
polynomial that corresponds to the row vector $\blda E$, we get from
the key equation of GRS decoding that
\begin{align}
    \label{eq:KE}
    \hLambda(x) \hsigma(x)
      &\equiv
         \hOmega(x)
           \quad (\mod \; x^{d-1}) \; .
\end{align}
Moreover,
\begin{align}
    \label{eq:gcd}
    \gcd
      (
        \hLambda(x), \hOmega(x)
      )
      &= 1 \; ,
\end{align}
\begin{alignat}{3}
    \label{eq:degrees1}
    \hspace{-0.25cm}
    \deg \hLambda(x)
      &= |\hJ|
      &\; \le \;& \Tau - \mu + 1
      &\; \le \;& \frac{d{-}\mu{-}\Rho}{2} \; , \\
    \label{eq:degrees2}
    \hspace{-0.25cm}
    \deg \hOmega(x)
      &< |K| + |\hJ|
      &\; \le \;& \Rho + \Tau - \mu + 1
      &\; \le \;& \frac{d{-}\mu{+}\Rho}{2} \; .
\end{alignat}
Multiplying both sides of~(\ref{eq:KE}) by $\lambda(x)$ we obtain
\begin{align*}
    \hLambda(x) \lambda(x) \hsigma(x)
      &\equiv
         \lambda(x) \hOmega(x)
           \quad (\mod \; x^{d-1}) \; .
\end{align*}
On the other hand, by~(P1) we have
\begin{align*}
    \lambda(x) \hsigma(x)
      &\equiv
         \homega(x)
           \quad (\mod \; x^{d-1}) \; ,
\end{align*}
and so, combining the last two congruences, we get
\begin{align*}
    \hLambda(x) \homega(x)
      &\equiv
         \lambda(x) \hOmega(x)
    \quad (\mod \; x^{d-1}) \; .
\end{align*}
Now, from~(P2) and~(\ref{eq:degrees1})--(\ref{eq:degrees2})
it follows that the degrees of the products on both sides of
the last congruence are less than $d - 1$.
Hence, this congruence is actually a polynomial equality:
\begin{align*}
    \hLambda(x) \homega(x)
      &= \lambda(x) \hOmega(x) \; .
\end{align*}
Thus, from~(\ref{eq:gcd}) we get that $\lambda(x)$ is divisible by
$\hLambda(x)$; in particular, $\lambda(x)$ is divisible by
$1 - \alpha_\ell x$.
Ranging now over all $\ell$ in $J$, we conclude that $\lambda(x)$ can
be written as $\Lambda(x) u(x)$ for some polynomial $u(x) \in F[x]$.
  
Finally, from~(P1) and the key equation we obtain
\begin{align*}
    \omega(y,x)
      &\equiv 
         \sigma(y,x)
         \lambda(x) \\
      &\equiv
         \sigma(y,x)
         \Lambda(x)
         u(x) \\
      &\equiv
         \Omega(y,x)
         u(x)
           \quad (\mod \; x^{d-1}) \; ,
\end{align*}
  from which the equality
$\omega(y,x) = \Omega(y,x) u(x)$ follows again by computing
degrees: from~(P2) we get
\begin{align*}
    \deg_x \omega(y,x)
      &<
         \frac{d+\mu+\Rho}{2} - 1
       \le
         d - 1
\end{align*}
and
\begin{align*}
      \lefteqn{
         \deg_x \big( \Omega(y,x) u(x) \big)
      } \makebox[6ex]{} \\
      &= \deg_x \Omega(y,x)
         +
         \big( \deg \lambda(x) - \deg \Lambda(x) \big) \\
      &< \deg \lambda(x)
         +
         \Rho \\
      &\le
         \frac{d+\mu+\Rho}{2} - 1 \\
      &\le
         d - 1 \; .
\end{align*}
This completes the proof.
\end{proof}

Fix $J$ and $K$ to be \emph{disjoint} subsets of $\Interval{n}$,
write $\Tau = |J|$ and $\Rho = |K|$, and let $E$ be
an $m \times n$ matrix of rank $\mu$ over $F$ such that
$J \subseteq \Support(E) \subseteq K \cup J$.  Define
\begin{align}
  \label{eq:decomposeS1}
  S^{(J)}
    = (E)_J \big( (H_\GRS)_J \big)^\transpose \; \phantom{.}
\end{align}
and
\begin{align}
  \label{eq:decomposeS2}
  S^{(K)}
     = (E)_K \big( (H_\GRS)_K \big)^\transpose \; .
\end{align}
Clearly, the syndrome array $S$ that corresponds to $E$ can be
decomposed into
\begin{align*}
  S
    &= S^{(J)} + S^{(K)} \; ,
\end{align*}
and after multiplying the respective bivariate polynomials by the
erasure-locator polynomial
$\Mu(x) = \prod_{j \in K}(1 - \alpha_j x)$, we get
\begin{align}
  \sigma(y,x)
    &\equiv
       S(y,x) \, \Mu(x) \nonumber \\
  \label{eq:sigma}
    &\equiv
       S^{(J)}(y,x) \, \Mu(x)
       +
       S^{(K)}(y,x) \, \Mu(x)
         \ (\mod \; x^{d-1}) \; .
\end{align}
Now, recall (from the key equation) that the coefficients of
$x^\Rho, x^{\Rho+1}, \ldots, x^{d-2}$ in $S^{(K)}(y,x) \, \Mu(x)$
are all zero, which means that the respective coefficients in
$\sigma(y,x)$ are equal to those in $S^{(J)}(y,x) \, \Mu(x)$.

\begin{figure}[ht!]
  \begin{myalgorithm}
    \small
    \textbf{Input:}
    \begin{itemize}
      
    \item Array $Y$ of size $m \times n$  over $F$.

    \item Set $K$ of size $\Rho$ of indexes of column erasures.

    \end{itemize}
    \textbf{Steps:}
    \begin{enumerate}
  
    \item
      \label{item:GRS1}
      Compute the $m \times (d{-}1)$ syndrome array
      \begin{align*}
        S
          &= Y H_\GRS^\transpose \; .
      \end{align*}
  
    \item
      \label{item:GRS2}
      Compute the modified syndrome array to be
      the unique $m \times (d{-}1)$
      matrix $\sigma$ that satisfies the congruence:
      \begin{align*}
        \sigma(y,x)
          &\equiv S(y,x) \, \Mu(x)
        \quad (\mod \; x^{d-1}) \; ,
      \end{align*}
      where
      \begin{align*}
        \Mu(x)
          &= \prod_{j \in K} (1 - \alpha_j x) \; .
      \end{align*}
      Let $\mu$ be the rank of the $m \times (d{-}1{-}\Rho)$
      matrix $\tS$ formed by the columns of $\sigma$ that are indexed by
      $\Interval{\Rho,d-1}$.
  
    \item
      \label{item:GRS3}
      Using the Feng--Tzeng algorithm, compute a polynomial
      $\lambda(x)$ of (smallest) degree
      $\Delta \le (d{+}\mu{-}\Rho)/2 - 1$ such that the following
      congruence is satisfied for some polynomial $\omega(y,x)$ with
      $\deg_x \omega(y,x) < \Rho + \Delta$:
      \begin{align*}
        \sigma(y,x) \lambda(x)
          &\equiv
             \omega(y,x)
               \quad (\mod \; x^{d-1}) \; .
      \end{align*}
      If no such $\lambda(x)$ exists or the computed $\lambda(x)$
      does not divide
      $\prod_{j \in \Interval{n}} (1 - \alpha_j x)$ then declare
      decoding failure and \textbf{Stop}.
  
    \item
      \label{item:GRS4}
      Compute the $m \times n$ error array $E$ by
        the following variant of Forney's formula
        for error values~\cite[p.~195]{Roth}:
      \begin{align*}
        E_j(y)
          &= \begin{cases}
               \frac{-\alpha_j \cdot \omega(y, \alpha_j^{-1})}
                 {\lambda'(\alpha_j^{-1})
                  \cdot
                  \Mu(\alpha_j^{-1})}
               & \textrm{if $\lambda(\alpha_j^{-1}) = 0$} \\
                 \frac{-\alpha_j \cdot \omega(y, \alpha_j^{-1})}
                 {\lambda(\alpha_j^{-1})
                  \cdot
                  \Mu'(\alpha_j^{-1})}
               & \textrm{if $j \in K$} \\
                 0 
               & \textrm{otherwise}
             \end{cases}
             \; ,
      \end{align*}
      where $(\cdot)'$ denotes formal differentiation.
  
    \end{enumerate}
    \textbf{Output:}
    \begin{itemize}
   
    \item Decoded array $Y - E$ of size $m \times n$.

    \end{itemize}
  \end{myalgorithm}
  \caption{Decoding of an $m$-level interleaving of a GRS code.
           (See Section~\ref{sec:GRS}.)}
  \label{fig:GRS}
\end{figure}

Let $\tS$ denote the $m \times (d{-}1{-}\Rho)$ matrix
\begin{align*}
  \tS
    &= (\sigma_{h,i})_{h \in \Interval{m},i\in \Interval{\Rho,d-1}} \; .
\end{align*}
It follows from~(\ref{eq:decomposeS1})--(\ref{eq:sigma}) that
\begin{align}
  \label{eq:Stildeagain}
  \tS
    &= (E)_J \bigl( (A_\Mu H_\GRS)_J \bigr)^\transpose \; ,
\end{align}
where $A_\Mu$ is a $(d{-}1{-}\Rho) \times (d{-}1)$ matrix over $F$
whose first row consists of the $\Rho + 1$ coefficients of
$\Mu(x)$ in decreasing order (padded with $d - 2 - \Rho$ zero entries), 
and each subsequent row is obtained from its predecessor by
a shift one position to the right (compare~(\ref{eq:Stildeagain})
with~(\ref{eq:Stilde})).  Hence, any column indexed by
$j \in J$ in the $(d{-}1{-}\Rho) \times \Tau$ matrix
$(A_\Mu H_\GRS)_J$ takes the form
\begin{align*}
  \alpha_j^\Rho
  \cdot
  \Mu(\alpha_j^{-1})
  \cdot
  \bigl(
    1 \; \alpha_j \; \alpha_j^2 \; \ldots \; \alpha_j^{d-2-\Rho}
  \bigr)^\transpose
\end{align*}
(where $\Mu(\alpha_j^{-1}) \ne 0$ since $K \cap J = \emptyset$).
Hence, under the assumption that $\Tau \le d - 1 - \Rho$
(which, in fact, holds whenever $2 \Tau + \Rho \le d + \mu - 2$)
we get that $\rank\bigl( (A_\Mu H_\GRS)_J \bigr) = \Tau$.
It now follows from~(\ref{eq:Stildeagain}) that
\begin{align}
  \label{eq:a}
  \colspan\bigl( (E)_J \bigr)
    &= \colspan(\tS)
\end{align}
(provided that $\Tau \le d - 1 - \Rho$);
equivalently, for every $\blda \in F^m$,
\begin{align*}
  \blda (E)_J
    &= \bldzero
  \quad \Longleftrightarrow \quad
  \blda \tS
     = \bldzero \; .
\end{align*}
In particular, $\rank\bigl( (E)_J \bigr) = \rank(\tS)$.  In fact,
every $m \times s$ sub-matrix of $\tS$ which consists of
$s \ge \Tau$ consecutive columns of $\tS$ has the same rank as $(E)_J$.

Given $\sigma(y,x) = \sum_{h \in \Interval{m}} \sigma_h(x) y^h$ and
the number of erasures $\Rho$, we can use
the Feng--Tzeng algorithm~\cite{FT}
(see also the related algorithms in~\cite{FT2, Wang, SB, ZW})
to find efficiently a polynomial
$\lambda(x) = \sum_{i=0}^\Delta \lambda_i x^i$ in $F[x]$ of
smallest degree $\Delta$ such that~(P1) is satisfied for some
$\omega(y,x) = \sum_{h \in \Interval{m}} \omega_h(x) y^h$ where $\deg_x
\omega(y,x) < \Rho + \Delta$.
In other words, for every $h \in \Interval{m}$,
the sequence $(\sigma_{h,i})_{i \in \Interval{\Rho,d-1}}$
satisfies the linear recurrence
\begin{align*}
  \sum_{i=0}^\Delta \lambda_i \sigma_{h,j-i}
    &= 0 \; ,
  \quad
  \Rho + \Delta
    \le
      j
    \le
      d - 2 \; .
\end{align*}
Under the assumption that $2 \Tau + \Rho \le d + \mu - 2$, we get from
Lemma~\ref{lem:GRS} that the polynomial $\lambda(x)$ found equals the
error-locator polynomial $\Lambda(x)$ (up to a normalization by a
constant). From the roots of $\lambda(x)$ one can then recover
the set $J$.

The decoding algorithm is summarized in Fig.~\ref{fig:GRS}.
Next, we analyze its complexity.
The syndrome computation step (Step~\ref{item:GRS1} in the figure)
can be carried out using $O(d m n)$ operations in $F$.
Step~\ref{item:GRS2} requires $O(d r m) = O(d^2 m)$ operations
to compute $\sigma(y,x)$ and
$O \bigl( d m \cdot \min(d,m) \bigr)$ operations
to compute the rank $\mu$.
The application of the Feng--Tzeng algorithm in
Step~\ref{item:GRS3} requires $O(d^2 m)$ operations, and, finally,
Step~\ref{item:GRS4} requires $O(d n)$ operations for
the Chien search and $O(d^2 m)$ operations for computing
the nonzero columns of $E$.
Overall, the decoding complexity amounts to
$O(d m n)$ operations for syndrome computation,
$O(d n)$ for the Chien search, and $O \bigl( d^2 m \bigr)$
for the remaining steps.

We note that
Lemma~\ref{lem:GRS} and the decoding algorithm in Fig.~\ref{fig:GRS} apply
also to the more general class of alternant codes over $F$, with $d$ now
standing for the designed minimum distance of the code.  Specifically, we
apply the lemma and the decoding algorithm to the underlying GRS code over the
appropriate extension field of $F$: that GRS code has minimum distance $d$ and
contains the alternant code as a sub-field sub-code.

\subsection{Application to Probabilistic Decoding}
\label{sec:probabilistic}

We next provide an application of the efficient decoding
algorithm of Fig.~\ref{fig:GRS} for the following channel model:
an $m \times n$ transmitted array over $F$ is subject to
at most a prescribed number $t$ of block errors (and no erasures),
such that the value (\ie, contents) of the block error in each
affected column is uniformly distributed over $F^m$, independently
of the other block-error values. (Note that in this formulation
of the channel, there is a probability of $1/q^m$ that
an affected column will in fact be error-free;
we have elected to define the channel this way in order to
simplify the analysis.)

We consider the decoding problem given that the $m \times n$ transmitted
array belongs to the code described in Section~\ref{sec:GRS},
namely, an $m$-level interleaving of an $[n,k,d{=}n{-}k{+}1]$ GRS code
over $F$. Let $J$ be the index set of  the columns that were affected
(possibly by an error-free block error), where $|J| \le t$,
and let $\mu$ be the rank of the error array $E$.
The decoder of Fig.~\ref{fig:GRS} will fail to decode
(or will decode incorrectly) \emph{only} when
the inequality in Lemma~\ref{lem:GRS} is violated, namely,
only when $\mu \le 2|J| - d + 1$.
Under the assumed statistical model on the error arrays,
it is easy to see that
\begin{align*}
\Prob \left\{ \rank(E) \le \mu \right\} 
  &= \Prob \left\{ \rank((E)_J) \le \mu \right\} \\
  & =
  q^{- (m-\mu)(|J|-\mu)} \cdot (1 + o(1))
\end{align*}
(see~\cite[p.~699]{MS}),
where $o(1)$ is an expression that goes to $0$
as $q \rightarrow \infty$.
Hence, when $m \ge d-1$,
the decoding failure probability of the algorithm
in Fig.~\ref{fig:GRS} is bounded from above,
up to a multiplicative factor $1 + o(1)$, by
\begin{align*}
  q^{- (m+d-1-2|J|)(d-1-|J|)} \le q^{- (m+d-1-2t)(d-1-t)} \; .
\end{align*}
For $m \ge d-1$, this bound is (considerably) better than
those given in~\cite{BKY} and~\cite{KSH}, and is comparable
to that in~\cite{SSB0, SSB} when $m$ is much larger than $d$.

\section{Main Code Construction}
\label{sec:main:code:construction}

\begin{table}
  \caption{Types of errors and erasures under consideration.}
  \label{table:error:and:erasure:types:1}

  \vspace{-0.25cm}

  \begin{center}
    \begin{tabular}{|c||c|c|}
      \hline
               & \parbox[c]{3cm}
                 {
                   \begin{center}
                     error
                   \end{center}
                 }
               & \parbox[c]{3cm}
                 {
                   \begin{center}
                     erasure
                   \end{center}
                 } \\
      \hline
      \hline
        block  & \parbox[c]{3cm}
                 {
                   \begin{center}
                     \textbf{(T1)} \\
                     column set $\JJ$ \\
                     $|\JJ| = \tau$
                   \end{center}
                 } 
               & \parbox[c]{3cm}
                 {
                   \begin{center}
                     \textbf{(T2)} \\
                     column set $\KK$ \\
                     $|\KK| = \rho$
                   \end{center}
                 } \\
      \hline
        symbol & \parbox[c]{3cm}
                 {
                   \begin{center}
                     \textbf{(T3)} \\
                     location set $\LL$ \\
                     $|\LL| = \vartheta$ 
                   \end{center}
                 }
               & \parbox[c]{3cm}
                 {
                   \begin{center}
                     \textbf{(T4)} \\
                     location set $\RR$ \\
                     $|\RR| = \varrho$
                   \end{center}
                 } \\
      \hline
    \end{tabular}
  \end{center}
\end{table}

We come now to the main code construction of this paper, namely the code
construction $\Code = (\code, H_\In)$ that was outlined in
Section~\ref{sec:introduction}. In particular,
Section~\ref{sec:code:definition} gives all
the details of the channel model and
the code construction, Section~\ref{sec:error:correction:capabilities}
presents the error and erasure correction capabilities of the code, and
Section~\ref{sec:examples} discusses a variety of examples based
on specific choices for the code $\code$ and
the matrix $H_\In$, and compares them with alternative code
constructions. (Decoding algorithms for $\Code$
will be discussed in Section~\ref{sec:decoding:algorithms}.)

\subsection{Channel Model and Code Definition}
\label{sec:code:definition}

We consider the following channel model. An $m \times n$ array $\Gamma$
over $F$ is stored (or transmitted), and $\Gamma$ is subject to
the following error and erasure types (see also
Table~\ref{table:error:and:erasure:types:1} and
Fig.~\ref{fig:memory:block:1}):
\begin{list}{}{\settowidth{\labelwidth}{\textup{(T4)}}%
               \settowidth{\leftmargin}{\textup{(T4..)}}}

\item[\textup{(T1)}]
\emph{Block errors:} a subset of columns in $\Gamma$ that
are indexed by $\JJ \subseteq \Interval{n}$ can be erroneous.

\item[\textup{(T2)}]
\emph{Block erasures:} a subset of columns in $\Gamma$ that
are indexed by $\KK \subseteq \Interval{n} \setminus \JJ$ can be erased.

\item[\textup{(T3)}]
\emph{Symbol errors:} a subset of entries in $\Gamma$
that are indexed by $\LL \subseteq \Interval{m} \times (\Interval{n}
\setminus (\KK \cup \JJ))$ can be erroneous.

\item[\textup{(T4)}]
\emph{Symbol erasures:} a subset of entries in $\Gamma$
that are indexed by
$\RR \subseteq \bigl( \Interval{m} \times (\Interval{n}
  \setminus \KK) \bigr) \setminus \LL$
can be erased.

\end{list}
Let the $m \times n$ matrix $\Epsilon$ over $F$ represent
the alterations that happen to $\Gamma$ over time
(or during transmission). Then the read out (or
received) message is given by the $m \times n$ matrix
\begin{align*}
  \Upsilon
    &= \Gamma + \Epsilon
\end{align*}
over $F$. Note that our formalism treats erasures like errors, with
the side information $\KK$ and $\RR$ telling us their location.
(Of course, the sets $\JJ$ and $\LL$ are not known to
the decoder \emph{a priori}.)

Write $\tau = |\JJ|$, $\rho = |\KK|$,
$\vartheta = |\LL|$, and $\varrho = |\RR|$.
The total number of symbol errors (resulting from error types~(T1)
and~(T3)) is at most $m \tau + \vartheta$ and the total number of symbol
erasures (resulting from erasure types~(T2) and~(T4)) is
$m \rho + \varrho$; hence,
we should be able to correct all error and erasure types
(T1)--(T4) (when occurring simultaneously) while using
a code of length $m n$ over $F$
with minimum distance at least
$m (2 \tau + \rho) + 2 \vartheta + \varrho + 1$.
However, such a strategy does not take into account the fact
that errors of type (T1)--(T2) are aligned across
the $m$ rows of the $m \times n$ array $\Gamma$.
The next construction is designed to take advantage of such 
an alignment.

\begin{definition}
  \label{def:main:code:construction}

Let $\code$ be a linear $[n,k,d]$ code over $F$, and let $H_\In$ be
an $m \times (m n)$ matrix over $F$ that satisfies
the following two properties for some positive integer $\delta$:
\begin{list}{}{\settowidth{\labelwidth}{\textup{(a)}}}
  
\item[\textup{(a)}]
Every subset of $\delta - 1$ columns in $H_\In$ is linearly independent 
(namely, $H_\In$ is a parity-check matrix of a linear
code over $F$ of length $m n$ and minimum distance at least $\delta$),
and
  
\item[\textup{(b)}]
writing
\begin{align*}
      H_\In
        &= \left(
             \begin{array}{c|c|c|c}
               H_0 & H_1 & \ldots & H_{n-1}
             \end{array}
           \right)
           \; ,
\end{align*}
with $H_0, H_1, \ldots, H_{n-1}$ being $m \times m$ sub-matrices of
$H_\In$, each $H_j$ is invertible over $F$.
  
\end{list}
Given $\code$ and $H_\In$, we define $\Code = (\code,H_\In)$ to be
the linear $[m n,m k]$ code over $F$ which consists of
all $m \times n$ matrices
\begin{align*}
    \Gamma
      &= \left(
           \begin{array}{c|c|c|c}
             \Gamma_0 & \Gamma_1 & \ldots & \Gamma_{n-1}
           \end{array}
         \right)
\end{align*}
over $F$ (where $\Gamma_j$ stands for column~$j$ of $\Gamma$) such
that each row in
\begin{align}
    \label{eq:HjGammaj}
    Z
      &= \left(
           \begin{array}{c|c|c|c}
             H_0 \Gamma_0 & H_1 \Gamma_1 & \ldots & H_{n-1} \Gamma_{n-1}
           \end{array}
         \right)
\end{align}
is a codeword of $\code$. \eqed
\end{definition}

One can view the code $\Code$ as a (generalized) concatenated code,
where the outer code is an $m$-level interleaving of $\code$,
such that an $m \times n$ matrix
\begin{align*}
  Z
    &= \left(
         \begin{array}{c|c|c|c}
           Z_0 & Z_1 & \ldots & Z_{n-1}
         \end{array}
       \right)
\end{align*}
over $F$ is an outer codeword if and only if each row in $Z$ belongs to
$\code$. Each column in $Z$ then undergoes encoding by
an inner encoder of rate one, where
the encoder of column~$j$ is given by the bijective mapping
$Z_j \mapsto H_j^{-1} Z_j$.

\subsection{Error Correction Capabilities}
\label{sec:error:correction:capabilities}

This subsection discusses what combinations of errors and erasures of
the types (T1)--(T4) can be handled by the code
$\Code = (\code, H_\In)$ that was
specified in Definition~\ref{def:main:code:construction}.

\begin{theorem}
\label{thm:combinederrors}

There exists a decoder for the code $\Code$ that correctly recovers
the transmitted array in the presence of errors of types~(T1)--(T4)
(which may occur simultaneously),
whenever $\tau \; (= |\JJ|)$, $\rho \; (= |\KK|)$,
$\vartheta \; (= |\LL|)$, and $\varrho \; (= |\RR|)$ satisfy
\begin{align*}
    2 \tau + \rho
      &\le
         d-2 \; , \\
    2 \vartheta + \varrho
      &\le
         \delta-1 \; .
\end{align*}
\end{theorem}

\begin{proof}
Using puncturing as in the proof of
Theorem~\ref{thm:generalizedmetzner}, it suffices to prove
the theorem for the case where $\KK$ is empty, \ie, there
are no erasure events of type (T2).

The proof is by contradiction. So, assume that $\Epsilon$ and
$\hEpsilon$ are two distinct $m \times n$ matrices that
correspond to error events of types~(T1), (T3), and~(T4),
with the respective sets $\JJ$, $\hJJ$, $\LL$,
and $\hLL$ satisfying
\begin{alignat*}{2}
    &2 \tau_{\max} 
      &&\le
         d-2 \; , \\
    &2 \vartheta_{\max} + \varrho
      &&\le
         \delta-1 \; ,
  \end{alignat*}
  where
  \begin{align*}
    \tau_{\max}
      &= \max \{ |\JJ|, |\hJJ| \} \; , \\
    \vartheta_{\max}
      &= \max \{ |\LL|, |\hLL| \}  \; ,
\end{align*}
and $\varrho = |\RR|$ (the symbol erasure set $\RR$ is
the same for both $\Epsilon$ and $\hEpsilon$).
We will have reached a contradiction once we have shown that
$\Epsilon - \hEpsilon$ is a codeword of $\Code$ if only if
$\Epsilon = \hEpsilon$.
  
Let $Z$ be the $m \times n$ matrix
\begin{align*}
    Z
      &= \left(
           \begin{array}{c|c|c|c}
             \!\! 
             H_0 (\Epsilon_0 {-} \hEpsilon_0)
             \! & \!
             H_1 (\Epsilon_1 {-} \hEpsilon_1)
             \! & \! 
             \ldots
             \! & \!
             H_{n-1} (\Epsilon_{n-1} {-} \hEpsilon_{n-1})
             \!\!
           \end{array}
         \right)
         \; .
\end{align*}
Observe that since the matrices $H_0, H_1, \ldots, H_{n-1}$ are all
invertible over $F$, we get that a column in $Z$ is zero if and only
if it is zero in $\Epsilon - \hEpsilon$; in particular, $Z = 0$ if
and only if $\Epsilon = \hEpsilon$.  Write
\begin{align*}
    \QQ = \Support(\Epsilon - \hEpsilon)
    \setminus (\JJ \cup \hJJ) \; ,
\end{align*}
namely, the set $\QQ$ indexes the columns of $\Epsilon - \hEpsilon$ that
contain errors of type~(T3) and (possibly part of) the erasures of
type~(T4).  For each $j \in \QQ$, denote by $w_j$ the number of nonzero
entries in $\Epsilon_j - \hEpsilon_j$ (that is,
$w_j = \bigl| \Support\bigl( (\Epsilon_j - \hEpsilon_j)^\transpose \bigr) 
\bigr|$).
The total number of nonzero entries in
$(\Epsilon - \hEpsilon)_\QQ$ satisfies:
\begin{align*}
    \sum_{j \in \QQ} w_j
      &\le
         |\LL \cup \hLL \cup \RR|
       \le
         2 \vartheta_{\max} + \varrho
       \le \delta-1 \; .
\end{align*}
Consider the respective columns in $Z$:
\begin{align*}
    Z_j
      &= H_j (\Epsilon_j - \hEpsilon_j) \; ,
           \quad j \in \QQ \; .
\end{align*}
Each $Z_j$ is a nontrivial linear combination of $w_j$ columns of
$H_\In$, and, obviously, for distinct indexes $j$
these combinations involve disjoint sets of columns of $H_\In$.
Recalling that every subset of
$\sum_{j \in \QQ} w_j \; (\le \delta - 1)$ columns in $H_\In$
is linearly independent, we thus get that the set of columns of
$(Z)_\QQ$ is linearly independent, that is,
\begin{align*}
    \rank(Z)
      &\ge
         \rank\bigl( (Z)_\QQ \bigr)
       = |\QQ| \; .
\end{align*}
On the other hand,
\begin{align*}
    |\Support(Z)|
      &\le
         |\JJ| + |\hJJ| + |\QQ|
       \le
         2 \tau_{\max} + |\QQ|
       \le d-2 + |\QQ|
\end{align*}
and, so,
\begin{align*}
    |\Support(Z)| - \rank(Z)
      &\le d - 2  \; .
\end{align*}
Hence, we conclude from Lemma~\ref{lem:generalizedmetzner} that each
row in $Z$ belongs to $\code$ only if $Z = 0$.
Equivalently, $\Epsilon - \hEpsilon$ belongs to $\Code$
if only if $\Epsilon = \hEpsilon$, as promised.
\end{proof}

\begin{remark}
  We draw the attention of the reader to the condition on $\tau$ and $\rho$ in
  Theorem~\ref{thm:combinederrors}, namely, that the expression $2 \tau +
  \rho$ be at most $d-2$, rather than the (more common) requirement that it be
  at most $d-1$.  It is this slightly stronger condition that, implicitly,
  provides the required redundancy for correcting the (additional) symbol
  errors and erasures.\eqed
\end{remark}

\begin{remark}
  Theorem~\ref{thm:combinederrors} indirectly implies a dependence of the
  correction capability of errors of type~(T3)--(T4) on the parameter $m$,
  which, in turn, is part of the specification of the errors of
  type~(T3)--(T4).  Specifically, the largest possible value for $\delta$ can
  be the minimum distance of any linear code of length $m n$ and redundancy
  $m$ over $F$.  Of course, one may re-arrange the array by grouping together
  non-overlapping sets of $s$ columns, for some integer $s > 0$, to form an
  $m' \times n'$ array where $m' = s m$ and $n' = n/s$ (assuming the latter
  ratio is an integer).  The block errors and the symbol errors will remain so
  also in this modified setting, except that the block errors will be more
  structured than just being phased with respect to the (new) parameter $m'$.
  If this additional structure is not taken into account, the code $\Code$ is
  bound to be sub-optimal.\eqed
\end{remark}

\begin{remark}
  In contrast to the previous remark, if $m$ is (much) larger than the
  redundancy needed from a linear code of length $m n$ and minimum distance
  $\delta$ as in Theorem~\ref{thm:combinederrors}, we may partition each
  column in the original array into $s$ new columns, thereby forming an $m'
  \times n'$ array, where $m' = m/s$ and $n' = s n$, such that $m'$ is (just)
  the redundancy required from a linear code of length $m n$ and minimum
  distance $\delta$.  Of course, this will increase the number of block errors
  by a factor of $s$, yet as long as $n'$ is sufficiently small to allow us to
  use a maximum-distance separable (MDS) code for $\code$ (as will be the case
  in the examples in Section~\ref{sec:examples}), we will obtain a gain in the
  redundancy: it will reduce from $(2 \tau + \rho) m + m$ to $(2 \tau + \rho)
  m + m'$. Thus, the code $\Code$ is suitable for (the rather practical)
  scenarios where the number of symbol errors is relatively small compared to
  the block-error length $m$.\eqed
\end{remark}

\begin{remark}
  One may speculate whether Theorem~\ref{thm:combinederrors} holds for the
  following more general definition of $\Code$: instead of requiring that each
  row of $Z$ in~(\ref{eq:HjGammaj}) be a codeword of $\code$, require that $Z$
  belong to a linear $[n,k,d]$ code over $\GF(q^m)$, where each column of $Z$
  is regarded as an element of the latter field with respect to some basis of
  that field over $F$.  It turns out that Theorem~\ref{thm:combinederrors}
  does \emph{not} hold for this more general setting, as shown by the
  following counterexample.

  Let $d = 2$ (\ie, $\tau = \rho = 0$). Assume that $n \le q^m - 1$ and let
  $C_0, C_1, C_2, \ldots, C_{n-1}$ be any $n$ distinct powers of an $m \times
  m$ companion matrix of some irreducible polynomial of degree~$m$ over
  $F$~\cite[p.~106]{MS}, where $C_0 = I$ (the $m \times m$ identity
  matrix). Then
  \begin{align*}
      \biggl\{
        Z = \left(
              \begin{array}{c|c|c|c}
                Z_0 & Z_1 & \ldots & Z_{n-1}
              \end{array}
            \right)
            \; : \;
            \sum_{j \in \Interval{n}} C_j Z_j = \bldzero
      \biggr\}
  \end{align*}
  is a linear $[n,n{-}1,2]$ code over $\GF(q^m)$ (with the columns of each $Z$
  being the codeword coordinates).  The respective code $\Code$ would then be
  written as
  \begin{align*}
      \biggl\{
        \Gamma = \left(
                   \begin{array}{c|c|c|c}
                     \Gamma_0 & \Gamma_1 & \ldots & \Gamma_{n-1}
                   \end{array}
                 \right)
                 \; : \;
                 \sum_{j \in \Interval{n}} C_j H_j \Gamma_j = \bldzero
      \biggr\} \; .
  \end{align*}
  If $\delta \ge 3$ then the $2m$ columns of $H_0$ and $H_1$ are all nonzero
  and distinct. Therefore, we can select $C_1 \ne I$ so that the first column
  (say) in $C_1 H_1$ equals the first column (say) in $H_0 = C_0 H_0$. Yet
  this means that there exists a nonzero $\Gamma \in \Code$ that contains only
  two nonzero entries. Therefore, $\Code$ cannot have a decoder that corrects
  all one-symbol error patterns.\eqed
\end{remark}

\subsection{Examples}
\label{sec:examples}

In this subsection, we consider various special choices for $\code$ and
$H_\In$ and demonstrate the properties of the resulting code
    $\Code = (\code, H_\In)$. Specifically,
    in Example~\ref{ex:optimal} below,
    we discuss complexity advantages,
    in an error-detection setting, that the construction
    $\Code$ can offer
    (with a suitable choice for $\code$ and $H_\In$)
    over MDS codes with the same length and size.
    In Examples~\ref{ex:GRS} and~\ref{ex:MDS},
    we show cases where $\Code$ is MDS (in fact, GRS),
    and, in contrast, we exhibit in Example~\ref{ex:nonMDS}
    a choice of parameters for which no MDS code can have the same
    length, size, and symbol--block correction capability as $\Code$.
    In Examples~\ref{ex:BCH}--\ref{ex:generalizedconcatenated},
    we demonstrate the advantages of the code $\Code$ over other
    existing alternatives for handling errors of types~(T1)--(T4),
    such as 
    concatenated codes and generalized concatenated (GC) codes.

\begin{example}
\label{ex:optimal}
We consider first the special case where $m n \le q + 1$.
Here, we can take $\code$ to be an MDS code over $F$ and $H_\In$
to be a parity-check matrix of an MDS code over $F$.
Under such circumstances we have $d = n - k + 1$ and $\delta = m + 1$,
which means that it suffices that the sizes $\tau$, $\rho$,
$\vartheta$, and $\varrho$ satisfy
\begin{align*}
    2 \tau + \rho
      &\le
         n - k - 1 \; , \\
    2 \vartheta + \varrho
      &\le
         m \; .
\end{align*}
The redundancy of $\Code$, being $m(n - k)$,
is then the smallest possible for this correction capability:
since the total number of symbol errors can be as large as
$m \tau + \vartheta$ and
the total number of symbol erasures is $m \rho + \varrho$, by
the Reiger bound~(\cite{LC}, \cite{PW}) we need a redundancy of
at least $m (2 \tau + \rho) + 2 \vartheta + \varrho$ symbols over $F$ in
order to be able to correct all error types (T1)--(T4).
Admittedly, the same performance of correction capability versus
redundancy can be achieved also by a single linear $[m n,m k]$ MDS code 
$\varcode$ over $F$ (which exists under
the assumption that $m n \le q + 1$).
However, as pointed out earlier, the use of such
a code $\varcode$ does not take into account the alignment of
error types (T1) and~(T2) across the rows of the received $m \times n$
array.  It is this alignment that allows $\Code$
to achieve the same correction capability using
a code $\code$, which is $m$ times shorter than $\varcode$.
While we still need for $\Code$ the parity-check matrix $H_\In$ of
an MDS code of length $m n$, the redundancy of the latter code needs to 
be only $m$, rather than $m(n - k)$.
  
To demonstrate the savings that $\Code$ may offer compared
to $\varcode$, consider the simple problem of verifying whether
a given $m \times n$ array $\Gamma$ belongs to the code
(namely, \emph{detecting} whether errors have occurred).
When using $\varcode$, we will regard $\Gamma$ as a vector of
length $m n$ over $F$ and the checking will be carried out through
multiplication by an $(m (n{-}k)) \times (m n)$ (systematic)
parity-check matrix of $\varcode$, thereby requiring up to
$2 m^2 k(n-k)$ operations (namely, additions and multiplications)
in $F$.  In contrast, when using $\Code$, we will first compute
the array $Z$ as in~(\ref{eq:HjGammaj}) while requiring less than
$2 m^2 (n-1)$ operations in $F$ (one of the matrices $H_j$ can be
assumed to be the identity matrix); then we will compute
the syndrome of each row of $Z$, for which we will need up to
$2 m k (n-k)$ operations in $F$.\eqed
\end{example}

\begin{example}
\label{ex:GRS}

Suppose that $m n \le q-1$ and select $\code$ to be
an $[n,k,d{=}n{-}k{+}1]$ GRS code $\code_\GRS$ over $F$ as in
Section~\ref{sec:GRS}.  For every $j \in \Interval{n}$, let
\begin{align*}
    H_j
      &= \left( \, 
           \beta_{\kappa,j}^h \,
         \right)_{h \in \Interval{m}, \kappa \in \Interval{m}} \; ,
\end{align*}
such that the elements $\beta_{\kappa,j}$ are distinct and nonzero
in $F$ for all $\kappa \in \Interval{m}$ and $j \in \Interval{n}$;
the respective matrix $H_\In = (H_j)_{j \in \Interval{n}}$ is then
a parity-check matrix of an $[m n, m(n{-}1), m{+}1]$ GRS code over $F$. 
Given any $m \times n$ matrix
$\Gamma = (\Gamma_{\kappa,j})_{\kappa\in\Interval{m},j\in \Interval{n}}$
over $F$, the entries of $H_j \Gamma_j$ 
are given by
\begin{align*}
    (H_j \Gamma_j)_h
      &= \sum_{\kappa \in \Interval{m}}
           \Gamma_{\kappa,j} \beta_{\kappa,j}^h \; ,
    \quad
    h \in \Interval{m} \; .
\end{align*}
  (Recall the definition of $\Gamma_j$ from 
  Definition~\ref{def:main:code:construction}.)
Hence, $\Gamma$ is in $\Code = (\code,H_\In)$ if and only if
\begin{align*}
    \sum_{\kappa \in \Interval{m}}
      \sum_{j \in \Interval{n}}
        \Gamma_{\kappa,j}
        \alpha_j^i
        \beta_{\kappa,j}^h
      &= 0 \; ,
           \quad
             h \in \Interval{m} \; ,
           \quad
             i \in \Interval{d-1} \; .
\end{align*}
(Note that if $\beta_{\kappa,j}$ depended only on $\kappa$, then $\Code$
could be seen as a two-dimensional shortening of a two-dimensional
cyclic code; see, for example~\cite{Sakata3}.)\eqed
\end{example}

\begin{example}
\label{ex:MDS}

We show that sometimes the construction $\Code$ in
Example~\ref{ex:GRS} is an MDS code.  Assume therein
that $m$ divides $q-1$ and that for every $j \in \Interval{n}$,
the multiplicative order of $\alpha_j$ divides $(q-1)/m$
(thus, each $\alpha_j$ has $m$ distinct $m$th roots in $F$).
For every $j \in \Interval{n}$, select
$\beta_{0,j}, \beta_{1,j}, \ldots, \beta_{m-1,j}$
to be the distinct roots of $\alpha_j$ in $F$.
Then $\Gamma \in \Code$ if and only if
\begin{align*}
    \sum_{\kappa \in \Interval{m}}
      \sum_{j \in \Interval{n}}
        \Gamma_{\kappa,j} \beta_{\kappa,j}^{h + m i}
      &= 0 \; ,
           \quad
             h \in \Interval{m} \; ,
           \quad
             i \in \Interval{d-1} \; .
  \end{align*}
The latter condition, in turn, is equivalent to $\Gamma$ being
a codeword of a GRS code of length $m n$ and redundancy
$(d - 1)m$ over $F$.\eqed
\end{example}

In contrast, the following example shows that sometimes
$\Code = (\code,H_\In)$ is not an MDS code, even when $\code$ is
MDS and $H_\In$ is a parity-check matrix of an MDS code.

\begin{example}
\label{ex:nonMDS}

Suppose that $q$ is a power of $4$ and take $m = 3$ and $n = (q + 2)/3$.
Select $H_\In$ to be a parity-check matrix of a $[q{+}2,q{-}1,4]$
triply-extended GRS code over $F$~\cite[p.~326]{MS} and $\code$ to be
any linear $[n,k]$ code over $F$.
Thus, $\Code$ is a linear $[q{+}2,3k]$ code over $F$.  It follows from
the already-proved range of the MDS conjecture that $\Code$, being
longer than $q + 1$, cannot be MDS when, say,
$2 \le k \le 1 + \frac{1}{6} \sqrt{q}$~\cite{ST}.\eqed
\end{example}

In fact, Example~\ref{ex:nonMDS} shows that there are choices of
$F$, $n$, $m$, $\tau$, and $\vartheta$ for which
the construction $\Code = (\code,H_\In)$ can be realized to correct any 
$\tau$ block errors and any $\vartheta$ symbol errors,
while, on the other hand, there are no codes over $F$ of the same
length and size as $\Code$ that can correct any
$2 \tau m + \vartheta$ symbol errors
(Example~\ref{ex:BCH} below presents
a larger range of parameters where this may happen).

In Examples~\ref{ex:BCH}--\ref{ex:generalizedconcatenated},
we make a running assumption that $n \le q$, in which case $\code$
can be taken as an MDS code, such as a (possibly extended) GRS code.
For the sake of simplicity, we will consider in these examples
only the block--symbol error-only case, \ie, no erasures are present.

\begin{example}
\label{ex:BCH}

Given positive $\tau$, $\vartheta$, $n$, and $F = \GF(q)$
(such that $2 \tau + 2 \le n \le q$), we take $H_\In$ to be
a parity-check matrix of a (possibly extended) shortened BCH code of
length $m n$ over $F$, where $m$ is determined by $\vartheta$, $n$,
and $q$ to satisfy the equality
\begin{align*}
    m
      &= 1
         + 
         \Bigl\lceil
           \frac{q{-}1}{q} \cdot (2 \vartheta - 1)
         \Bigr\rceil
         \cdot
         \Bigl\lceil
           \log_q (m n)
         \Bigr\rceil
\end{align*}
(so $m n$ may be larger than $q$; see~\cite[p.~260]{Roth}).  The code
$\code$ is taken as a (possibly extended) $[n,k,d]$ GRS code over $F$
where $d = 2\tau+2$. The overall redundancy of
$\Code = (\code,H_\In)$ is then
\begin{align}
    \label{eq:BCH}
    (2 \tau + 1) m 
      &= 2 \tau m
         + 
         1
         +
         \Bigl\lceil
           \frac{q{-}1}{q} \cdot (2 \vartheta - 1)
         \Bigr\rceil
         \cdot
         \Bigl\lceil
           \log_q (m n)
         \Bigr\rceil \; .
\end{align}
The first term, $2 \tau m$, on the right-hand side of~(\ref{eq:BCH}) is
the smallest redundancy possible if one is to correct any $\tau$
block errors of length $m$. The remaining term therein is
the redundancy (or an upper bound thereof)
of a BCH code that corrects any $\vartheta$ symbol errors
over $F$.  In comparison, a shortened BCH code of length $m n$ over $F$ 
that corrects any $\tau m + \vartheta$ symbol errors may have
redundancy as large as
\begin{align*}
    1
    +
    \Bigl\lceil
      \frac{q{-}1}{q}
      \cdot
      \bigl(
        2 (\tau m + \vartheta) - 1
      \bigr)
    \Bigr\rceil
    \cdot
    \Bigl\lceil
      \log_q (m n)
    \Bigr\rceil \; .
\end{align*}
It can be verified that the latter expression is larger
than~(\ref{eq:BCH}) when $m n > q \ge 4$.\eqed
\end{example}

\begin{example}
\label{ex:concatenated}

We compare the performance of $\Code$ with that of a concatenated code
$\varcode$ constructed from a linear $[m,\Kin,\Din]$ inner code over $F$
and a linear $[n,\Kout,\Dout]$ outer code over $\GF(q^\Kin)$
(where $n \le q$).
By the Singleton bound, we can bound the redundancy of $\varcode$ from
below by
\begin{align}
    m n - \Kin \Kout
      &\ge
         (\Dout \! - \! 1)m
         +
         (\Din \! - \! 1)n
         -
         (\Dout \! - \! 1)(\Din \! - \! 1) 
    \nonumber \\
    \label{eq:redcont1}
      &= (\Dout \! - \! 1)(m \! + \! 1)
         -
         n
         +
         \Din (n \! - \! \Dout \! + \! 1) \; .
\end{align}
As we mentioned already in Section~\ref{sec:related:work},
any error pattern of up to $\tau$ block errors and up to $\vartheta$
symbol errors can be correctly decoded, whenever
\begin{align}
    \label{eq:ZZ}
    2 \vartheta + 1 
      &\le
         \Din (\Dout - 2 \tau) \; .
\end{align}
Since in our setting the values $\tau$ and $\vartheta$ are prescribed,
we can minimize~(\ref{eq:redcont1}) over $\Din$ and $\Dout$,
subject to the inequality~(\ref{eq:ZZ}).
Specifically, we define $\Delta = \Dout - 2 \tau$
and, from~(\ref{eq:ZZ}), we can express $\Din$ in terms of $\Delta$
as $\Din = \left\lceil (2 \vartheta + 1)/\Delta \right\rceil$, in which
case~(\ref{eq:redcont1}) becomes
\begin{align}
\label{eq:redcont2}
    &\hskip-0.6cm
    \Delta (m+1)
    +
    \left\lceil
      \frac{2 \vartheta + 1}{\Delta}
    \right\rceil
    (n-2\tau+1-\Delta)
      \nonumber \\
    &\hskip-0.65cm
     + (2\tau-1)(m+1)
     - n \\
    &\label{eq:redcont3}
     \ge
       \Delta (m+1)
       +
       \frac{(2 \vartheta + 1)(n-2\tau+1)}{\Delta}
         \nonumber \\
     &\quad\,
      -
      (2\vartheta + 1)
      -
      (n-2\tau+1)
      +
      (2\tau-1)m \; . \hskip-0.25cm
\end{align}
The minimum of~(\ref{eq:redcont3}) over (a real) $\Delta$ is attained at
\begin{align*}
    \Delta_{\min} = \sqrt{\frac{(2 \vartheta + 1)(n-2\tau+1)}{m+1}} \; .
\end{align*}
Yet we need to take into account that both $\Din$ and $\Delta$ are
positive integers.
  
\emph{Case 1: $n-2 \tau < \frac{m+1}{2\vartheta+1}$.}  In this range (of
very small $n$) we have $\Delta_{\min} \le 1$ and, so, we take
$\Delta = 1$; namely, for this range, it is best to let
the inner code handle (exclusively) the symbol errors,
and then the outer code is left to correct the block errors.
For this range, the expression~(\ref{eq:redcont2}) becomes
\begin{align*}
    2 \tau m + 2 \vartheta (n - 2\tau) \; ,
\end{align*}
which is never larger than the smallest possible redundancy,
$(2 \tau + 1)m$, of $\Code$.
Notice, however, that when $H_\In$ in $\Code$ can be taken as
a parity-check matrix of an MDS code,
then $m = 2 \vartheta$ and therefore this range is empty.
  
\emph{Case 2: $n-2\tau \ge (m+1) \cdot (2\vartheta+1)$.}
In this range we have $\Delta_{\min} > 2\vartheta + 1$ and, so,
the expression~$\lceil \cdot \rceil$ (for $\Din$)
in~(\ref{eq:redcont2}) equals~$1$; namely, for this range,
it is best to regard the symbol errors as block errors.
Therefore, we take $\Delta = 2\vartheta + 1$ and the
expression~(\ref{eq:redcont2}) becomes
\begin{align*}
    2 (\tau + \vartheta) m \; ,
\end{align*}
which is larger than the redundancy of $\Code$ whenever $\vartheta > 0$
(assuming that $\code$ in $\Code$ is an MDS code).
  
\emph{Case 3:
$\frac{m+1}{2\vartheta+1} \le n - 2\tau < (m+1) \cdot (2 \vartheta+1)$.}
In this range, we plug the expression for $\Delta_{\min}$
into~(\ref{eq:redcont3}), resulting in the following lower bound
expression on the redundancy of $\varcode$:
\begin{align*}
    & 2 \sqrt{(2 \vartheta + 1)(n-2\tau+1)(m+1)}
      - (2\vartheta + 1) \\
    & - (n-2\tau+1) + (2\tau-1)m \; .
\end{align*}
This lower bound is greater than the redundancy of $\Code$ whenever
\begin{align}
    4 (2 \vartheta + 1)&(n-2\tau+1)(m+1) \nonumber \\
    \label{eq:redcont4}
      &> \Bigl(
           2m 
           + 
           (2\vartheta + 1)
           +
           (n-2\tau+1)
         \Bigr)^2 \; .
\end{align}
Since the left-hand side of~(\ref{eq:redcont4}) is a cubic expression
while the right-hand side is only quadratic,
the inequality~(\ref{eq:redcont4}) is expected to hold for
a range of parameters of interest, e.g., when
$m$, $\tau$, and $\vartheta$ scale linearly with $n$.\eqed
\end{example}

\begin{example}
\label{ex:generalizedconcatenated}

In many cases, the redundancy of $\Code$ is smaller even than that
of the generalized concatenated (GC) code construction defined
through conditions~(G1)--(G2) in Section~\ref{sec:related:work}.
Referring to the notation therein, we first note that
if $\Dout_v > 2\tau  + 1$, then the contribution of condition~(G2)
alone to the redundancy is already at least
$(\Dout_v{-}1)m \ge (2\tau+1)m$.
Hence, we assume that $\Dout_v = 2\tau + 1$, 
in which case, from~(\ref{eq:generalizedconcatenated}),
we get that $\Rin_v \ge \Din_v-1 \ge 2\vartheta$.
Thus, condition~(G2) induces a redundancy of at least $2\tau m$,
and condition~(G1) adds a redundancy of at least
\begin{align}
\label{eq:generalizedconcatenatedredundancy}
\sum_{i=1}^v (\Rin_i - \Rin_{i-1}) (\Dout_{i-1} - 1 - 2\tau) \; .
\end{align}
In Appendix~\ref{sec:generalizedconcatenated},
we show that the expression~(\ref{eq:generalizedconcatenatedredundancy})
is bounded from below by
\begin{align}
\label{eq:generalizedconcatenatedlowerbound}
(2\vartheta+1) \ln \vartheta + 2 \gamma \cdot \vartheta + O(1) \; ,
\end{align}
where $\gamma$ is Euler's constant
(approximately $0.5772$)~\cite[p.~264]{GKP}.
Taking $\code$ in $\Code$ as an MDS code,
the redundancy of $\Code$ is then smaller
than that of GC codes
(with the same correction capabilities)
whenever~(\ref{eq:generalizedconcatenatedlowerbound})
exceeds $m$.\eqed
\end{example}

\section{Decoding Algorithms for the \\ 
               Main Code Construction}
\label{sec:decoding:algorithms}

We now discuss a variety of decoders for
the code $\Code = (\code, H_\In)$ that was specified in
Definition~\ref{def:main:code:construction},
for the case where the constituent code $\code$ is a GRS code.
Section~\ref{sec:polynomialtime} presents
a polynomial-time decoding algorithm for error and erasures of types
(T1)--(T4) as far as they are correctable as guaranteed by
Theorem~\ref{thm:combinederrors}, and provided that
the code parameters $n$, $d$, $\delta$, $\rho$, and $\varrho$
satisfy a certain inequality (see Theorem~\ref{thm:GS} below).
Then, we consider the construction $\Code$ as in Example~\ref{ex:GRS}
(where $\code$ is a GRS code and $H_\In$ is a parity-check matrix of
a GRS code) and present some more specialized decoders for
this construction. Namely,
Section~\ref{sec:decoding:t1:t2:t4} discusses a decoder that handles
errors and erasures of types (T1), (T2), (T4), but not of type (T3),
and Section~\ref{sec:specialconfigurations} introduces a decoder that
handles errors and erasures of types~(T1), (T2), (T4),
and some combinations of errors of type~(T3),
  including the case where there are at most three errors
  of type~(T3).%
  \footnote{%
      A small number of errors of type~(T3) is
      a realistic assumption for memory storage
      applications that use \emph{scrubbing}, i.e., memory storage
      applications where a background task periodically inspects
      the memory for errors and corrects them if necessary.
      Such a background task helps avoid the accumulation of
      errors between the time that a program writes and reads
      a certain memory location.
  }
As of yet, we do not have an efficient decoder that corrects all error
patterns that satisfy the conditions of
Theorem~\ref{thm:combinederrors}
(even for the construction of Example~\ref{ex:GRS},
excepting certain special cases such as Example~\ref{ex:MDS}).

Assume that $m$, $\tau$, $\rho$, $\vartheta$, and $\varrho$ scale
linearly with $n$.
If $\Code$ is replaced by a GRS code (if such a code exists) then
the decoding complexity scales linearly with $(n^2)^2 = n^4$. 
  One of the main purposes of defining the code
  $\Code = (\code, H_\In)$ is the potential existence of
  a decoding algorithm whose complexity does not scale higher than
  $n^3$, as is the case for the special choices of
  $\Code = (\code, H_\In)$ and decoders in
  Sections~\ref{sec:decoding:t1:t2:t4}
  and~\ref{sec:specialconfigurations}.

\subsection{Polynomial-Time Decoding Algorithm}
\label{sec:polynomialtime}

Example~\ref{ex:MDS} demonstrates one particular instance of
the construction $\Code = (\code,H_\In)$ for which
the decoder guaranteed by Theorem~\ref{thm:combinederrors} has
an efficient implementation (simply because in
this case the code $\Code$ is a GRS code). In this section, we exhibit
a much wider range of instances for which decoding can be carried
out in polynomial-time complexity.

Specifically, we consider the case where the code $\code$ is a GRS code
over $F$ (and, so, $n \le q-1$), and $H_\In$ is
an \emph{arbitrary} $m \times (m n)$ matrix over $F$ that satisfies
the two properties (a)--(b) in
Definition~\ref{def:main:code:construction}.
The columns of $m \times n$ arrays will be regarded as elements of
the extension field $\GF(q^m)$ (according to some basis of
$\GF(q^m)$ over $F$).
When doing so, the matrix $Z$ in~(\ref{eq:HjGammaj}) can be seen as
a codeword of a GRS code $\code'$ over $\GF(q^m)$, where $\code'$
has the same code locators $(\alpha_j)_{j \in \Interval{n}}$ as
$\code$
(this observation was used, for example, in~\cite{SSB2},
and more recently in~\cite{GX}).

Let $\Gamma \in \Code$ be the transmitted $m \times n$ array and let
$\Upsilon$ be the received $m \times n$ array, possibly corrupted by
$\tau$ errors of type~(T1) and $\vartheta$ errors of type~(T3),
where $\tau \le (d/2) - 1$ and $\vartheta \le (\delta-1)/2$.
We first compute an $m \times n$ array
\begin{align*}
  Y
    &= \left(
         \begin{array}{c|c|c|c}
       H_0 \Upsilon_0 & H_1 \Upsilon_1 & \ldots & H_{n-1} \Upsilon_{n-1}
         \end{array}
       \right)
       \; ,
\end{align*}
where $Y$ contains $\tau + \vartheta \le (d+\delta-3)/2$ erroneous
columns. Regarding now $Y$ as a corrupted version of
a codeword of $\code'$, we can apply a \emph{list decoder} for $\code'$ 
to $Y$.  Such a decoder returns a list of up to a prescribed number
$\List$ of codewords of $\code'$, and the returned list is guaranteed
to contain the correct codeword, provided that the number of erroneous
columns in $Y$ does not exceed the \emph{decoding radius} of $\code'$.
In our decoding, we will use the polynomial-time list decoder due
to Guruswami and Sudan~\cite{GS}
(for variations of the algorithm that reduce its complexity,
see~\cite{KV} and~\cite{Wu}).
For any GRS code of length $n$ and
minimum distance $d$ (over any field) and for any prescribed list size
$\List$, their decoder will return the correct codeword as long as
the number of errors does not exceed $\Ceil{n \GS_\List(d/n)} - 1$,
where $\GS_\List(d/n)$ is the maximum over
$s \in \{ 1, 2, \ldots, \List \}$ of the following expression:
\begin{align*}
  \GS_{\List,s}(d/n)
    &= 1
       -
       \frac{s+1}{2(\List+1)}
       -
       \frac{\List}{2s}
       \left(
         1 - \frac{d}{n}
       \right)
\end{align*}
(see~\cite[Chapter~9.5]{Roth}).  Thus, if $\List$ is such that
\begin{align}
  \label{eq:decodingradius}
  n \GS_\List(d/n)
    &\ge
       (d+\delta-1)/2 \; ,
\end{align}
then the returned list will contain the correct codeword
\begin{align*}
  Z
    &= \left(
         \begin{array}{c|c|c|c}
           H_0 \Gamma_0 & H_1 \Gamma_1 & \ldots & H_{n-1} \Gamma_{n-1}
         \end{array}
       \right)
\end{align*}
of $\code'$. For each array $Z'$ in the list we can compute
the respective array in $\Code$,
\begin{align*}
  \Gamma'
    &= \left(
         \begin{array}{c|c|c|c}
          H_0^{-1} Z'_0 & H_1^{-1} Z'_1 & \ldots & H_{n-1}^{-1} Z'_{n-1}
         \end{array}
       \right)
       \; ,
\end{align*}
and it follows from the proof Theorem~\ref{thm:combinederrors} that only
one such computed array $\Gamma'$---namely, the transmitted array
$\Gamma$---can correspond to an error pattern of up to
$(d/2) - 1$ block errors and up to $(\delta-1)/2$
symbol errors. Finding that array can be done simply
by checking each computed $Z'$ against the received array $\Upsilon$.

\begin{remark}
While the decoding scheme that we have just outlined makes
essential use of an efficient list decoder for the $m$-level
interleaving of $\code$, nothing is assumed about $H_\In$
beyond properties~(a)--(b) in
Definition~\ref{def:main:code:construction}.
In particular, nothing is assumed about the decoding
complexity of the code over $F$ that is defined by
the parity-check matrix $H_\In$.\eqed
\end{remark}

Our decoding scheme can be generalized to handle also erasures of
types~(T2) and~(T4) by applying a list decoder for
the GRS code obtained by puncturing $\code'$ on
the columns that are affected by erasures: this translates into
replacing $d$ by $d - \rho - \varrho$ (assuming that the latter value
is positive).

The next theorem characterizes a range of parameters for
which~(\ref{eq:decodingradius}) holds for some polynomially-large list
size $\List$ (and, thus, $\Code$ can be decoded in polynomial time).

\begin{theorem}
\label{thm:GS}
For $\Code = (\code,H_\In)$ such that $\code$ is a GRS code over $F$,
the decoder guaranteed by Theorem~\ref{thm:combinederrors} can be
implemented by a polynomial-time algorithm, whenever
\begin{align}
    \label{eq:d/n}
    d - \rho - \varrho
      &\ge
         2 \sqrt{\delta (n-\rho-\varrho)}
         -
         \delta
\end{align}
(or, simply, whenever $d \ge 2 \sqrt{\delta n} - \delta$
in case $\rho = \varrho = 0$).
\end{theorem}

\begin{proof}
We will assume in the proof that $\rho = \varrho = 0$;
the general case follows by observing that any puncturing of $\code'$
on $\rho + \varrho$ positions results in a GRS code of
length $n - \rho - \varrho$ and minimum distance $d - \rho - \varrho$.
  
Our proof will be complete once we identify
a poly\-nomially-large $\List$ 
for which~(\ref{eq:decodingradius}) holds.
We take $\List$ to be such that
\begin{align}
    \label{eq:GS}
    \frac{d}{n}
      &\ge
         2 \sqrt{\frac{\delta-1}{n}}
         -
         \frac{\delta-1}{n}
         +
         \frac{2}{\List+1} \; .
\end{align}
It readily follows from~(\ref{eq:d/n}) that there exists such
an $\List$ which is at most quadratic in $n$.

Define $s$ to be
\begin{align*}
    s
      &= \List 
         -
         \left\lceil
           \List 
           \cdot 
           \sqrt{\frac{\delta-1}{n}}
         \right\rceil \; .
\end{align*}
  From~(\ref{eq:GS}) we have
\begin{align*}
    \frac{d}{n}
      &\ge
         \frac{s}{\List-s}
         \cdot
         \frac{\delta-1}{n}
         +
         \frac{\List+1-s}{\List+1} \; ,
\end{align*}
which can also be rewritten as
\begin{align*}
    \frac{d}{n}
      &\ge
         \frac{2s}{\List-s}
         \left(
           \frac{\delta-1}{2n}
           +
           \frac{\List}{2s}
           +
           \frac{s+1}{2(\List+1)}
           -
           1
         \right) \; .
\end{align*}
Multiplying both sides by $(\List-s)/(2s)$ and rearranging terms yields
  \begin{align*}
    1
    -
    \frac{s+1}{2(\List+1)}
    -
    \frac{\List}{2s}
      \left(
        1 - \frac{d}{n}
      \right)
      &\ge
         \frac{d}{2n}
         +
         \frac{\delta-1}{2n} \; ,
\end{align*}
which is equivalent to
\begin{align*}
    \GS_{\List,s}(d/n)
      &\ge
         \frac{d + \delta - 1}{2n} \; .
\end{align*}
This, in turn, implies~(\ref{eq:decodingradius}).
\end{proof}

The range of parameters in~(\ref{eq:d/n}) may potentially be
increased in light of a recent result of Guruswami and Xing
on list decoding of interleaved GRS codes~\cite{GX}.

\subsection{Decoding of Errors and Erasures of Type (T1), (T2), (T4),
                              but not of Type (T3)}
\label{sec:decoding:t1:t2:t4}

In this section, we present an efficient decoder for the code $\Code$
when constructed as in Example~\ref{ex:GRS}, for the special case
where $\vartheta = |\LL| = 0$ (no errors of type~(T3)).%
  \footnote{%
      This special case has also been considered in~\cite{SSB0},
      yet under the setting of Section~\ref{sec:probabilistic}, namely,
      where the decoding algorithm may fail
      with a (controlled) positive probability.
    }

An $m \times n$ matrix $\Gamma$ is transmitted and
an $m \times n$ matrix
\begin{align*}
  \Upsilon
    &= \Gamma + \Epsilon
\end{align*}
is received, where
$\Epsilon = (\varepsilon_{\kappa,j})_{\kappa \in \Interval{m},
         j \in \Interval{n}}$
is an $m \times n$ error matrix, with
$\JJ \; (\subseteq \Interval{n})$
(respectively, $\KK \; (\subseteq \Interval{n})$) indexing
the columns in which block errors (respectively,
block erasures) have occurred, and
$\RR \; (\subseteq \Interval{m} \times \Interval{n})$ is
a nonempty set of positions where symbol erasures
have occurred.%
  \footnote{%
      When performing arithmetic operations on $\Upsilon$,
      we assume that the erased entries in the array are preset to some
      arbitrarily-selected elements of $F$, whereas the sets $\KK$ and
      $\RR$ are provided as side information. Thus, $\Epsilon$ is also
      an array over $F$.
  }
We assume that $d$,
$\tau \; (= |\JJ|)$, and $\rho \; (= |\KK|$) satisfy
\begin{align}
  \label{eq:bounderasure}
  2 \tau + \rho
    &\le d - 2
\end{align}
and that $\varrho \; (= |\RR|)$ satisfies
\begin{align*}
  0
    &< \varrho
     \le 
       m \; .
\end{align*}
Define
\begin{align*}
  Y
    &= \left(
         \begin{array}{c|c|c|c}
       H_0 \Upsilon_0 & H_1 \Upsilon_1 & \ldots & H_{n-1} \Upsilon_{n-1}
         \end{array}
       \right)
\end{align*}
and
\begin{align*}
  \lefteqn{
    E = (e_{h,j})_{h \in \Interval{m}, j \in \Interval{n}}
  } \makebox[15ex]{} \\
    &= \left(
         \begin{array}{c|c|c|c}
       H_0 \Epsilon_0 & H_1 \Epsilon_1 & \ldots & H_{n-1} \Epsilon_{n-1}
         \end{array}
       \right)
       \; .
\end{align*}
Clearly,
\begin{align*}
  Y
    &= Z + E \; ,
\end{align*}
where $Z$ is given by~(\ref{eq:HjGammaj}). In particular, every row
in $Z$ is a codeword of $\code_\GRS$.

Next, write
$\RR = \{ (\kappa_\ell,j_\ell) \}_{\ell \in \Interval{\varrho}}$.
  For each $\ell \in \Interval{\varrho}$, define the following
  univariate polynomial (of degree $\varrho-1$)
\begin{align}
  \label{eq:Beta}
  \Beta^{(\ell)}(y)
    &= \sum_{i \in \Interval{\varrho}}
         \Beta_i^{(\ell)} y^i \\
    &= \prod_{(\kappa,j) \in \RR \setminus\{(\kappa_\ell,j_\ell)\}}
         \frac{1 - \beta_{\kappa,j} y}
              {1 - \beta_{\kappa,j}\beta_{\kappa_\ell,j_\ell}^{-1}} \; ,
\end{align}
and let $\blde^{(\ell)} = (e_j^{(\ell)})_{j \in \Interval{n}}$ denote
row~$\varrho-1$ of the $(m{+}\varrho{-}1) \times n$ matrix
$\Beta^{(\ell)}(y) E(y,x)$
  (where we recall the definition of $E(y,x)$ from Section~\ref{sec:notation}).
We have
\begin{align*}
  \Support(\blde^{(\ell)})
    &\subseteq
       \JJ
       \cup
       \KK
       \cup
       \{ j_\ell \} \; ,
         \quad  \ell \in \Interval{\varrho} \; .
\end{align*}
Indeed, the contribution of a symbol erasure at position
$(\kappa,j)$ in $\Epsilon$ to the column $E_j(y)$ of $E(y,x)$ is
an additive term of the form
\begin{align*}
  \varepsilon_{\kappa,j}
  \cdot
  \uu_m(y;\beta_{\kappa,j})
    &= \varepsilon_{\kappa,j}
       \cdot
       \frac{1 - (\beta_{\kappa,j} y)^m}
            {1 - \beta_{\kappa,j} y}
\end{align*}
(where we recall the definition of $\uu_m(\cdot;\cdot)$ from
Section~\ref{sec:notation});
so, if $(\kappa,j) \ne (\kappa_\ell,j_\ell)$ then the product
\begin{align*}
  \Beta^{(\ell)}(y)
  {\cdot}
  \varepsilon_{\kappa,j}
  {\cdot}
  \frac{1 - (\beta_{\kappa,j} y)^m}{1 - \beta_{\kappa,j} y} =
  \varepsilon_{\kappa,j}
  {\cdot}
  \frac{\Beta^{(\ell)}(y)}{1 - \beta_{\kappa,j} y}
  {\cdot}
  \bigl( 1 - (\beta_{\kappa,j} y)^m \bigr)
\end{align*}
is a polynomial in which the powers
$y^{\varrho-1}, y^\varrho, \ldots, y^{m-1}$ have zero coefficients.

Now, 
   for each $\ell \in \Interval{\varrho}$,
every row in the $(m{+}\varrho{-}1) \times n$ array
$Z^{(\ell)}(y,x) = \Beta^{(\ell)}(y) Z(y,x)$ is a codeword of
$\code_\GRS$. Therefore, by applying a decoder for
$\code_\GRS$ to row~$\varrho-1$ of $Z^{(\ell)}$ with $\rho+1$ erasures
indexed by $\KK \cup \{ j_\ell \}$, we should be able to
decode the vector $\blde^{(\ell)}$, based on our
assumption~(\ref{eq:bounderasure}).

It follows from the definition of~$\blde^{(\ell)}$ that for every $j \in
\Interval{n}$,
\newcommand{\RB}{\raisebox{0ex}[2.5ex][1ex]{}}
\begin{align}
\label{eq:filtering}
  \left(
    \begin{array}{@{\!}c@{\!}}
     \RB e_j^{(0)} \\ \RB e_j^{(1)} \\
     \RB \vdots    \\ \RB e_j^{(\varrho-1)}
    \end{array}
  \right)
    &= \left(
         \begin{array}{@{\!}cccc@{\!}}
     \RB
     \Beta_0^{(0)} & \Beta_1^{(0)} & \ldots & \Beta_{\varrho-1}^{(0)} \\
     \RB
     \Beta_0^{(1)} & \Beta_1^{(1)} & \ldots & \Beta_{\varrho-1}^{(1)} \\
     \RB
     \vdots        & \vdots        & \vdots & \vdots                  \\
     \RB
     \Beta_0^{(\varrho-1)} & \Beta_1^{(\varrho-1)} & \ldots &
     \Beta_{\varrho-1}^{(\varrho-1)} \\
         \end{array}
       \right)
       \left(
         \begin{array}{@{\!}c@{\!}}
           \RB e_{\varrho-1,j} \\ \RB e_{\varrho-2,j} \\
           \RB \vdots \\ \RB e_{0,j}
         \end{array}
       \right)
       \; .
\end{align}
In particular,
\begin{align*}
  e_{j_\ell}^{(\ell)}
    &= \sum_{i \in \Interval{\varrho}}
         B_i^{(\ell)}
           \sum_{\kappa \,:\, (\kappa,j_\ell) \in \RR}
             \varepsilon_{\kappa,j_\ell}
             \beta_{\kappa,j_\ell}^{\varrho-1-i} \\
    &= \sum_{\kappa \,:\, (\kappa,j_\ell) \in \RR}
         \varepsilon_{\kappa,j_\ell}
         \beta_{\kappa,j_\ell}^{\varrho-1}
         \Beta^{(\ell)}(\beta_{\kappa,j_\ell}^{-1}) \\
    &= \varepsilon_{\kappa_\ell,j_\ell}
         \beta_{\kappa,j_\ell}^{\varrho-1} \; .
\end{align*}
Ranging over all $\ell \in \Interval{\varrho}$,
we are able to recover the erasures 
in $\Epsilon$ at the positions $\RR$. Namely,
\begin{align*}
  \varepsilon_{\kappa_\ell,j_\ell}
    &= e_{j_\ell}^{(\ell)}
       \beta_{\kappa_\ell,j_\ell}^{1-\varrho} \; ,
         \quad \ell \in \Interval{\varrho} \; .
\end{align*}
This, in turn, allows us to eliminate
the symbol erasures from $E$.

Fig.~\ref{fig:allerasuredecoding} summarizes the decoding algorithm
of a combination of errors of type~(T1), (T2), and~(T4).
The complexity of Step~\ref{item:1} is
$O \bigl( (d+m)m n \bigr)$ operations in $F$ (see the discussion that
precedes Example~\ref{ex:GRS}).  Step~\ref{item:2} requires
$O(d \rho \varrho)$ operations.
Each iteration in Step~\ref{item:3} requires $O(d \varrho)$ operations
(for Step~\ref{item:3a}), $O(d^2)$ operations (for Step~\ref{item:3b}), 
and $O(d m)$ operations (for Step~\ref{item:3c}),
totaling to $O \bigl( d(d+m) \varrho \bigr)$ for Step~\ref{item:3}.
Step~\ref{item:4} requires $O(d^2 m)$ operations to compute
the error-locator and error-evaluator polynomials,
and an
additional $O(dn)$ for the Chien search.
Finally, Step~\ref{item:5} requires $O(d m^2)$ operations.
To summarize, the decoding complexity amounts to
$O \bigl( (d+m)m n \bigr)$ operations for syndrome computation,
$O ( d n )$ for the Chien search,
and $O \bigl( d(d+m)m \bigr)$ for the remaining steps.

\begin{figure}[pht!]
  \begin{myalgorithm}
    \small
    \textbf{Input:}
    \begin{itemize}

    \item Array $\Upsilon$ of size $m \times n$ over $F$.

    \item Set $\KK$ of indexes of column erasures.

    \item Set $\RR = \{ (\kappa_\ell,j_\ell) \}_{\ell \in
        \Interval{\varrho}}$ of positions of symbol erasures.

    \end{itemize}
    \textbf{Steps:}
    \begin{enumerate}

    \item
      \label{item:1}
      Compute the $m \times (d{-}1)$ syndrome array
      \begin{align*}
        S 
          &= \left(
               \begin{array}{c|c|c|c}
       H_0 \Upsilon_0 & H_1 \Upsilon_1 & \ldots & H_{n-1} \Upsilon_{n-1}
               \end{array}
             \right)
             H_\GRS^\transpose \; .
      \end{align*}
    
    \item
      \label{item:2}
      Compute the modified syndrome array to be the unique
      $\varrho \times (d{-}1)$ matrix $\sigma$ that satisfies
      the congruence
      \begin{align*}
        \sigma(y,x)
          &\equiv
             S(y,x)
             \prod_{j \in \KK}
               (1 - \alpha_j x)
                 \quad (\mod \; \{ x^{d-1}, y^\varrho\}) \; .
      \end{align*}

    \item
      \label{item:3}
      For every $\ell \in \Interval{\varrho}$ do:
      \begin{enumerate}
      \itemsep0.75ex
      \item
        \label{item:3a}
        Compute row~$\varrho-1$ in the unique $\varrho \times (d{-}1)$
        matrix $\sigma^{(\ell)}$ that satisfies the congruence
        \begin{align*}
          \sigma^{(\ell)}(y,x)
            &\equiv
               \Beta^{(\ell)}(y) \, 
               \sigma(y,x) \, 
               (1 - \alpha_{j_\ell} x) \\
            &\quad\quad\quad\quad\quad\quad\quad\quad\quad
                 (\mod \; \{ x^{d-1}, y^\varrho \}) \; ,
        \end{align*}
        where $\Beta^{(\ell)}(y)$ is as in~(\ref{eq:Beta}).

    \item
      \label{item:3b}
      Decode $e_{j_\ell}^{(\ell)}$
      (\ie, entry~$j_\ell$ in $\blde^{(\ell)}$)
      by applying a decoder for $\code_\GRS$ using row~$\varrho-1$ in
      $\sigma^{(\ell)}$ as syndrome and assuming that columns indexed by
      $\KK \cup \{ j_\ell \}$ are erased.  Compute
      $\varepsilon_{\kappa_\ell,j_\ell} = e_{j_\ell}^{(\ell)} \cdot
      \beta_{\kappa_\ell,j_\ell}^{1-\varrho}$.

    \item
      \label{item:3c}
      Update the received array $\Upsilon$ and the syndrome array $S$ by
      \begin{align*}
        \Upsilon(y,x)
          &\leftarrow
             \Upsilon(y,x)
             -
             \varepsilon_{\kappa_\ell,j_\ell}
               {\cdot}
               x^{j_\ell}
               y^{\kappa_\ell} \\
        S(y,x)
          &\leftarrow
             S(y,x)
             -
             \varepsilon_{\kappa_\ell,j_\ell}
               {\cdot}
               \uu_{d-1}(x; \alpha_{j_\ell})
               {\cdot}
               \uu_m(y; \beta_{j_\ell}) \; .
      \end{align*}

    \end{enumerate}

    \item
      \label{item:4}
      For every $h \in \Interval{m}$, apply a decoder for $\code_\GRS$
      using row~$h$ of $S$ as syndrome and assuming that columns
      indexed by $\KK$ are erased.
      Let $E$ be the $m \times n$ matrix whose rows are the
      decoded error vectors for all $h \in \Interval{m}$.

    \item
      \label{item:5}
      Compute the error array
      \begin{align*}
        \Epsilon
          &= \left(
               \begin{array}{c|c|c|c}
             H_0^{-1} E_0 & H_1^{-1} E_1 & \ldots & H_{n-1}^{-1} E_{n-1}
               \end{array}
             \right)
             \; .
      \end{align*}

    \end{enumerate}
    \textbf{Output:}
    \begin{itemize}

    \item Decoded array $\Upsilon - \Epsilon$ of size $m \times n$.

    \end{itemize}
  \end{myalgorithm}
  \caption{Decoding of errors and erasures of type (T1), (T2), (T4), but
    not of type (T3). (See Section~\ref{sec:decoding:t1:t2:t4}.)}
  \label{fig:allerasuredecoding}
\end{figure}

\subsection{Decoding of Errors and Erasures of Type (T1), (T2), (T4),
                          and with Restrictions on Errors of Type (T3)}
\label{sec:specialconfigurations}

In this section, we consider the decoding of $\Code$ when constructed as in
Example~\ref{ex:GRS}, under certain assumptions on the set $\LL$, namely,
under some restrictions on the symbol error positions (errors of type~(T3)).
These restrictions always hold when $|\LL| \le 3$ and $d$ is sufficiently
large.

Specifically, we consider the case where each column, except possibly
for one column, contains at most one symbol error.
The general strategy will be to locate the positions of these errors,
thereby reducing to the case considered in
Section~\ref{sec:decoding:t1:t2:t4}. We use the same notation as in
that section, except that the set $\LL$ is not necessarily empty and
that (for reasons of simplicity) the set $\RR$ is empty. As in
Section~\ref{sec:decoding:t1:t2:t4}, the number $\tau$ of
block errors and the number $\rho$ of
block erasures satisfy $2 \tau + \rho \le d-2$.

When $\vartheta = |\LL| > 0$, we write
$\LL = \{ (\kappa_\ell,j_\ell) \}_{\ell \in \Interval{\vartheta}}$,
and assume that that there exists
a $w \in \Interval{\vartheta}$ such that the values
$j_0, j_1, \ldots, j_w$ are all distinct, while
$j_w = j_{w+1} = \cdots = j_{\vartheta-1}$.  Furthermore,
$\vartheta$ and $w$ should satisfy the inequalities
\begin{align}
  \label{eq:vartheta:1}
  \vartheta
    &\le
       \frac{m}{2} \; , \\
  \label{eq:vartheta:2}
  w + \tau + \rho
    &\le
       d-2 \; .
\end{align}
(While the inequality in~(\ref{eq:vartheta:1}) is already part of the
requirements in Theorem~\ref{thm:combinederrors}, we need 
the inequality in~\eqref{eq:vartheta:2} 
so that~(\ref{eq:a}) will hold.  Specifically, the inequality
in~(\ref{eq:vartheta:2}) says that the number of erroneous columns does not
exceed $d-1$. Observe that the inequality $2 \tau + \rho \le d-2$ and the
inequality in~(\ref{eq:vartheta:1}) together imply~(\ref{eq:vartheta:2})
whenever $m \le d-\rho$.)

Without any loss of generality, we will also assume that
$\varepsilon_{\kappa_\ell,j_\ell} \ne 0$ for every
$\ell \in \Interval{\vartheta}$.
The set $\{ j_\ell \}_{\ell \in \Interval{w+1}}$ will
be denoted hereafter by $\LL'$.
When $\vartheta = 0$, we formally define $w$
to be $0$ and $\LL'$ to be the empty set.

Let the modified syndrome $\sigma$ be
the unique $m \times (d{-}1)$ matrix that satisfies
\begin{align*}
  \sigma(y,x)
    &\equiv
       S(y,x)
       \cdot
       \prod_{j \in \KK}
         (1 - \alpha_j x)
           \quad (\mod \; x^{d-1}) \; ,
\end{align*}
and let $\tS$ be the $m \times (d{-}1{-}\rho)$ matrix formed by
the columns of $\sigma$ that are indexed by $\Interval{\rho,d-1}$.
We recall from~(\ref{eq:a}) that
$\mu = \rank(\tS) = \rank\bigl( (E)_{\JJ \cup \LL'}
\bigr)$.

If $\mu \ge 2w + 2$, then we can regard the columns that are indexed
by $\LL'$ as full block errors
(namely, errors of type~(T1)), and the conditions
of Lemma~\ref{lem:GRS} will still be satisfied, namely, we will have
\begin{align*}
  2(\tau + w + 1) + \rho
    &\le
       d + \mu - 2 \; .
\end{align*}
Therefore, we assume from now on that $\mu \le 2 w + 1$.

By~(\ref{eq:a}) we get that for
every $j \in \JJ \cup \LL'$, column~$E_j$ belongs to $\colspan(\tS)$.
In particular, this holds for $j \in \LL' \setminus \{ j_w \}$,
in which case $E_j$ (in polynomial notation) takes the form
\begin{align*}
  E_j(y)
    &= \varepsilon_{\kappa,j}
       \cdot
       \uu_m(y;\beta_{\kappa,j}) \; .
\end{align*}

Let the row vectors $\blda_0, \blda_1, \ldots, \blda_{m-\mu-1}$ form
a basis of the dual space of $\colspan(\tS)$,
and for every $i \in \Interval{m-\mu}$,
let $a_i(y)$ denote the polynomial of degree less than $m$
with coefficient vector $\blda_i$;
we can further assume that this basis is in echelon form, \ie,
$\deg a_0(y) < \deg a_1(y) < \ldots < \deg a_{m-\mu-1}(y) < m$.
This, in turn, implies that the degree of
$a(y) = \gcd(a_0(y), a_1(y), \ldots, a_{m-\mu-1}(y))$ satisfies
\begin{align*}
  \deg a(y)
    &\le
       \mu \; ,
\end{align*}
namely, $a(y)$ has at most $\mu \; (\le 2 w + 1)$ distinct roots in $F$.
Now, it is easy to see that for every $\xi \in F$,
the column vector $( \xi^h)_{h \in \Interval{m}}$
(also represented as $\uu_m(y;\xi)$) belongs to $\colspan(\tS)$
(and, hence, to $\colspan(E)_{\JJ \cup \LL'}$), if and only if
$\xi$ is a root of $a(y)$. In particular,
$\beta_{\kappa_\ell,j_\ell}$ is a root of $a(y)$ for every
$\ell \in \Interval{w}$. We denote by $\RR$ the root subset
\begin{align}
\label{eq:RR}
  \RR
    &= \bigl\{
         (\kappa,j) 
         \,:\, 
         a(\beta_{\kappa,j}) = 0 
       \bigr\} \; ,
\end{align}
and define the polynomial $\Alpha(y)$ by
\begin{align}
  \label{eq:Alpha}
  \Alpha(y) 
    &= \sum_{i=0}^\eta
         \Alpha_i y^i
     = \prod_{(\kappa,j) \in \RR}
         (1 - \beta_{\kappa,j} y) \; ,
\end{align}
where $\eta = |\RR|$.

Consider the $(m{-}\eta) \times n$ matrix
$\hE = (\he_{h,j})_{h \in \Interval{m-\eta},j\in\Interval{n}}$
which is formed by the rows of
$\Alpha(y) E(y,x)$ that are indexed by $\Interval{\eta,m}$.
Specifically,
\begin{align*}
  \he_{h,j} =
  \sum_{i=0}^\eta \Alpha_i \, e_{h+\eta-i,j} \; ,
  \quad
  h \in \Interval{m-\eta} \; ,
  \quad
  j \in \Interval{n}
\end{align*}
(compare with~(\ref{eq:filtering})).  Respectively, let $\hS$ be
the $(m{-}\eta) \times (d{-}1{-}\rho)$ matrix formed by
the rows of $\Alpha(y) \tS(y,x)$ that are indexed by
$\Interval{\eta,m}$.  It readily follows that $\hE_{j_\ell}(y) = 0$
for $\ell \in \Interval{w}$ and that
\begin{align}
\label{eq:jw}
  \hE_{j_w}(y)
    = \!\! \sum_{\ell \in \Interval{w,\vartheta}}
        \left(
          \varepsilon_{\kappa_\ell,j_w}
          \beta_{\kappa_\ell,j_w}^\eta
          \Alpha(\beta_{\kappa_\ell,j_w}^{-1})
        \right)
        \cdot
        \uu_{m-\eta}(y;\beta_{\kappa_\ell,j_w})
        \; .
\end{align}
Observe that the number of summands on
the right-hand side of~(\ref{eq:jw}) is $\vartheta - w$,
and that number is bounded from above by
$m - 2w - 1 \le m - \mu \le m - \eta$.
This means that $\hE_{j_w}(y) = 0$ if and only if
$\Alpha(\beta_{\kappa_\ell,j_w}^{-1}) = 0$ for all $\ell \in
\Interval{w,\vartheta}$.  We also recall that
\begin{align}
  \label{eq:hatS}
  \rank(\hS)
    &= \rank\bigl( (\hE)_{\JJ \cup \{ j_w \}} \bigr)
     = \mu - \eta \; .
\end{align}

Next, we distinguish between the following three cases.

\emph{Case~1: $\eta = \mu$.}
By~(\ref{eq:hatS}) we must have $\hE_{j_w}(y) = 0$,
which is equivalent to having
$\Alpha(\beta_{\kappa_\ell,j_\ell}^{-1}) = 0$
for all $\ell \in \Interval{\vartheta}$.
Thus, assuming this case, we have
$\LL \subseteq \RR$, and the decoding problem then reduces to the one
discussed in Section~\ref{sec:decoding:t1:t2:t4}.

\emph{Case~2: $\eta = \mu - 1$.}
If $\hE_{j_w}(y) = 0$ then $\LL \subseteq \RR$.
Otherwise, it follows from~(\ref{eq:hatS}) that each column
in $\hS$ must be a scalar multiple of $\hE_{j_w}$.  The entries of
$\hE_{j_w}$, in turn, form a sequence that satisfies
the (shortest) linear recurrence
\begin{align*}
  \Beta(y)
    &= \sum_{i=0}^{|\RR'|}
         \Beta_i y^i
     = \prod_{(\kappa,j) \in \RR'}
         (1 - \beta_{\kappa,j} y) \; ,
\end{align*}
where
\begin{align*}
  \RR'
    &= \Bigl\{
         (\kappa_\ell,j_w)
         \, : \,
         \textrm{$\ell \in \Interval{w,\vartheta}$ and
                 $\Alpha(\beta_{\kappa_\ell,j_w}^{-1}) \ne 0$}
       \Bigr\} \; .
\end{align*}
Indeed, this recurrence is uniquely determined, since
the number of entries in $\hE_{j_w}$,
which is $m - \eta = m - \mu + 1 \ge m - 2w$, is at least
twice the degree $|\RR'| \; (\le \vartheta - w)$ of $\Beta(y)$.  The
recurrence can be computed efficiently from any nonzero column of $\tS$
using Massey's algorithm
  for finding the shortest linear feedback shift register capable of
  generating a prescribed finite sequence of symbols~\cite{Massey}
  (cf.~Berlekamp--Massey algorithm as in, e.g., \cite{Roth}).
We now have $\LL \subseteq \RR \cup \RR'$, where
\begin{align*}
  |\RR \cup \RR'|
    &= |\RR| + |\RR'| \\
    &\le
       \eta + \vartheta - w \\
    &\le
       2w + \vartheta - w \\
    &= \vartheta + w \\
    &< m \; .
\end{align*}
So the decoding problem again reduces to 
that in Section~\ref{sec:decoding:t1:t2:t4}.

\emph{Case~3: $\eta \le \mu - 2$.}
 If $\hE_{j_w}(y) = 0$ then (again) $\LL \subseteq \RR$.
Otherwise, the conditions of Lemma~\ref{lem:GRS} hold with respect
to $\hE$ and to the matrix $\hZ$ formed by the rows of
$\Alpha(y) Z(y,x)$ indexed by $\Interval{\eta,m}$
(each such row is a codeword of $\code_\GRS$).
Hence, we can decode $\hE$.  Next, we observe from~(\ref{eq:jw})
that for $j = j_w$, the vector $\hE_j(y)$ can be seen as
a syndrome of the column vector
\begin{align*}
  \Epsilon_j^*(y)
    &= \sum_{\kappa \in \Interval{m} \, : \,
               \Alpha(\beta_{\kappa,j}^{-1}) \ne 0}
         \varepsilon_{\kappa,j} y^\kappa
\end{align*}
with respect to the following $(m{-}\eta) \times m$ parity-check matrix
of a GRS code:
\begin{align}
\label{eq:Hj1}
  H_\GRS^{(j)}
    &= \left(
         v_{\kappa,j} \beta_{\kappa,j}^h
       \right)_{h \in \Interval{m-\eta}, \kappa \in \Interval{m}}
       \; ,
\end{align}
where
\begin{align}
\label{eq:Hj2}
  v_{\kappa,j}
    &= \left\{
         \begin{array}{ccl}
           \beta_{\kappa,j}^\eta \Alpha(\beta_{\kappa,j}^{-1})
           && \textrm{if $\Alpha(\beta_{\kappa,j}^{-1}) \ne 0$} \\
           1
           && \textrm{otherwise}
         \end{array}
       \right.
       \; .
\end{align}
And since the Hamming weight of $\Epsilon_j^*$ is at most
$\vartheta - w < (m - \eta)/2$, we can decode $\Epsilon_j^*$
uniquely from $\hE_j$ (again, under the running assumption that
$j = j_w$).  Thus, for every $\kappa$ such that
$\Alpha(\beta_{\kappa,j}^{-1}) \ne 0$, we can recover the error value
$\varepsilon_{\kappa,j}$ and subtract it from the respective entry of
$\Upsilon$, thereby making $\RR$
a superset of the remaining symbol errors.
The problem though is that we do not know the index $j_w$.  Therefore,
we apply the above process to \emph{every} nonzero column in $\hE$
with index $j \not\in \KK$. A decoding failure means that $j$
is certainly not $j_w$, and a decoding success for $j \ne j_w$ will
just cause us to incorrectly change already corrupted-columns in
$\Upsilon$, without introducing new erroneous columns.
We can then proceed with the decoding of $\Upsilon$ as in
Section~\ref{sec:decoding:t1:t2:t4}.

\begin{figure}
  \begin{myalgorithm}
    \small
    \textbf{Input:}
    \begin{itemize}

    \item Array $\Upsilon$ of size $m \times n$  over $F$.

    \item Set $\KK$ of indexes of column erasures.

    \end{itemize}
    \textbf{Steps:}
    \begin{enumerate}
      \itemsep0ex

    \item
      \label{item:dec1}
      Compute the $m \times (d{-}1)$ syndrome array
      \begin{align*}
        S
          &= \left(
               \begin{array}{c|c|c|c}
       H_0 \Upsilon_0 & H_1 \Upsilon_1 & \ldots & H_{n-1} \Upsilon_{n-1}
               \end{array}
             \right)
             H_\GRS^\transpose \; .
      \end{align*}

    \item
      \label{item:dec2}
      Compute the $m \times (d{-}1{-}\rho)$ matrix $\tS$ formed by
      the columns of
      $S(y,x) \prod_{j \in \KK} (1 - \alpha_j x)$ that are indexed
      by $\Interval{\rho,d-1}$.  Let $\mu = \rank(\tS)$.

    \item
      \label{item:dec3}
      (Attempt to correct assuming $|\LL'| \le \mu/2$.)  Apply
      Steps~\ref{item:GRS3}--\ref{item:GRS4} in Fig.~\ref{fig:GRS}
      (with $K = \KK$) to the modified syndrome array $\sigma(y,x)$,
      to produce an error array $E$.  If decoding is successful, go to
      \textbf{Step~\ref{item:dec8}}.

    \item
      \label{item:dec4}
      \begin{enumerate}
        \itemsep0.75ex

      \item
        \label{item:dec4a}
        Compute the greatest common divisor $a(y)$ of a basis of
        the left kernel of $\tS$.

      \item
        \label{item:dec4b}
        Compute the set $\RR$ and the polynomial $\Alpha(y)$ as
        in~(\ref{eq:RR})--(\ref{eq:Alpha}). Let $\eta = |\RR|$.

      \item
        \label{item:dec4c}
        Compute the $(m{-}\eta) \times (d{-}1{-}\rho)$ matrix $\hS$
        formed by the rows of $\Alpha(y) \tS(y,x)$ that are indexed by
        $\Interval{\eta,m}$.

      \end{enumerate}

  \item
    \label{item:dec5}
    If $\eta = \mu - 1$ then do:
    \begin{enumerate}
      \itemsep0.75ex

    \item
      \label{item:dec5a}
      Compute the shortest linear recurrence $\Beta(y)$ of
      any nonzero column in $\hS$.

    \item
      \label{item:dec5b}
      Compute the set
      \begin{align*}
        \RR'
          &= \left\{
               (\kappa,j)
               \,:\,
               \textrm{$\Alpha(\beta_{\kappa,j}^{-1}) \ne 0$ and
                       $\Beta(\beta_{\kappa,j}^{-1}) = 0$}
             \right\}.
       \end{align*}

    \item
      \label{item:dec5c}
      If $|\RR'| = \deg \Beta(y)$ and $|\RR'| \le m - \eta$
      then update $\RR \leftarrow \RR \cup \RR'$.

    \end{enumerate}

  \item
    \label{item:dec6}
    Else if $\eta \le \mu - 2$ then do:
    \begin{enumerate}
      \itemsep0.75ex

    \item
      \label{item:dec6a}
      Apply Steps~\ref{item:GRS2}--\ref{item:GRS4} in
      Fig.~\ref{fig:GRS} (with $K = \KK$) to the syndrome array $\hS$,
      to produce an error array $\hE$.

    \item
      \label{item:dec6b}
      For every index $j \not\in \KK$ of a nonzero column of
      $\hE$ do:
      \begin{enumerate}
        \itemsep0ex

      \item
        \label{item:(j)}
        Apply a decoder for the GRS code with the parity-check
        matrix~$H_\GRS^{(j)}$ as in~(\ref{eq:Hj1})--(\ref{eq:Hj2}), with
        $\hE_j$ as syndrome, to produce an error vector $\Epsilon_j^*$.

      \item If decoding in Step~\ref{item:(j)} is successful then
        let $E_j^* = H_j \Epsilon_j^*$ and update
        $\Upsilon_j \leftarrow \Upsilon_j - \Epsilon_j^*$ and
        $S(y,x) \leftarrow
                 S(y,x) - E_j^*(y) \cdot \uu_{d-1}(x; \alpha_{j_\ell})$.

      \end{enumerate}

    \end{enumerate}

  \item
    \label{item:dec7}
    Apply Steps~\ref{item:2}--\ref{item:4} in
    Fig.~\ref{fig:allerasuredecoding} to $S$, $\KK$, and $\RR$,
    to produce an error array $E$.

  \item
    \label{item:dec8}
    Compute the error array
    \begin{align*}
      \Epsilon
        &= \left(
             \begin{array}{c|c|c|c}
             H_0^{-1} E_0 & H_1^{-1} E_1 & \ldots & H_{n-1}^{-1} E_{n-1}
             \end{array}
           \right)
           \; .
    \end{align*}
    \end{enumerate}
    \textbf{Output:}
    \begin{itemize}

    \item Decoded array $\Upsilon - \Epsilon$ of size $m \times n$.

    \end{itemize}
  \end{myalgorithm}
  \caption{Decoding of errors and erasures of type (T1), (T2), (T4),
    and with restrictions on errors of type (T3). (See
    Section~\ref{sec:specialconfigurations}.)
    For the sake of simplicity, we
    assume that there are no erasures of type (T4).}
  \label{fig:specialconfigurations}
\end{figure}

Fig.~\ref{fig:specialconfigurations} presents
the implied decoding algorithm of a combination of errors of type~(T1), 
(T2), and~(T3), provided that the type-(T3) errors satisfy
the assumptions laid out at the beginning of this section;
as said earlier, these assumptions hold when $m \le d - \rho$ and
the number of type-(T3) errors is at most~$3$.
Steps~\ref{item:dec1}--\ref{item:dec3},
\ref{item:dec6a}, and~\ref{item:dec7}--\ref{item:dec8}
in Fig.~\ref{fig:specialconfigurations} are
essentially applications of steps in
Figs.~\ref{fig:GRS} and~\ref{fig:allerasuredecoding}.
Next, we analyze the complexity of the remaining steps
in Fig.~\ref{fig:specialconfigurations},
starting with Step~\ref{item:dec4a}.
A basis in echelon form of the left kernel of $\tS$ can be found
using $O(d^2 m)$ operations in $F$, and from this basis we can compute
$a(y)$ using $m$ applications of Euclid's algorithm, amounting to
$O(d m^2)$ operations.
The set $\RR$ can then be found in Step~\ref{item:dec4b} via a
Chien search, requiring $O \bigl( m n \cdot \min(d,m) \bigr)$
operations, followed by $O(d^2 m)$ operations to compute
the matrix $\hS$ in Step~\ref{item:dec4c}.
The complexity of Step~\ref{item:dec5} is dictated by
Step~\ref{item:dec5b} therein which, with a
Chien search,
can be implemented using $O(m^2 n)$ operations.
Finally, Step~\ref{item:dec6b} requires
$O \bigl( d (d+m) m \bigr)$ operations in $F$.
In summary, the decoding complexity of the algorithm in 
Fig.~\ref{fig:specialconfigurations}
amounts to $O \bigl( (d+m)m n \bigr)$ operations
for syndrome computation and the Chien search,
and $O \bigl( d(d+m)m \bigr)$ operations for the other steps.

\section*{Acknowledgment}

The authors thank Erik Ordentlich for helpful discussions.

\appendices

\section{Analysis for Example~\protect\ref{ex:generalizedconcatenated}}
\label{sec:generalizedconcatenated}

We derive the lower bound~(\ref{eq:generalizedconcatenatedlowerbound})
on the expression~(\ref{eq:generalizedconcatenatedredundancy}):

\begin{align*}
\lefteqn{
\sum_{i=1}^v (\Rin_i - \Rin_{i-1}) (\Dout_{i-1} - 1 - 2\tau)
} \makebox[8ex]{} \\
&\stackrel{(\ref{eq:generalizedconcatenated})}{\ge}
\sum_{i=1}^v (\Rin_i - \Rin_{i-1}) 
\left(
\left\lceil \frac{2\vartheta+1}{\Din_{i-1}} \right\rceil 
- 1 \right) \\
&\ge
\sum_{i=1}^v (\Rin_i - \Rin_{i-1}) 
\left(
\left\lceil \frac{2\vartheta+1}{\Rin_{i-1}+1} \right\rceil 
-1 \right) \\
&=
\sum_{i=1}^v
\sum_{j=\Rin_{i-1}+1}^{\Rin_i}
\left(
\left\lceil \frac{2\vartheta+1}{\Rin_{i-1}+1} \right\rceil 
-1 \right) \\
&\ge
\sum_{j=1}^{\Rin_v}
\left(
\left\lceil \frac{2\vartheta+1}{j} \right\rceil - 1 \right)
\end{align*}
\begin{align*}
&\ge
\sum_{j=1}^{2\vartheta}
\left(
\left\lceil \frac{2\vartheta+1}{j} \right\rceil - 1 \right) \\
&=
\sum_{j=1}^\vartheta
\left\lceil \frac{2\vartheta+1}{j} \right\rceil \; ,
\end{align*}
where the penultimate step follows from $\Rin_v \ge 2 \vartheta$.
The lower bound~(\ref{eq:generalizedconcatenatedlowerbound})
immediately follows by the known expression for
harmonic sums~\cite[p.~264]{GKP}.

\end{document}